\documentclass[amsthm]{elsart}
\usepackage{amssymb}
\usepackage{amsthm}
\usepackage{amsmath}
\usepackage{amsbsy}
\usepackage{color}
\usepackage{array}
\usepackage{epsfig}
\usepackage{graphicx}
\usepackage{pmat}
\usepackage{yjsco}
\usepackage{natbib}

\usepackage{lscape}

\theoremstyle{plain}
\newtheorem{theorem}{Theorem}
\newtheorem{proposition}[theorem]{Proposition}
\newtheorem{lemma}[theorem]{Lemma}

\newcounter{algor}
\theoremstyle{definition}
\newtheorem{definition}[theorem]{Definition}
\newtheorem{algori}[algor]{Algorithm}

\theoremstyle{remark}
\newtheorem{rmk}[theorem]{Remark}

\newcommand{\QQ}{\mathbb Q}
\newcommand{\RR}{\mathbb R}
\newcommand{\CC}{\mathbb C}
\newcommand{\AAA}{\mathbb A}
\newcommand{\XX}{\mathbb X}
\newcommand{\YY}{\mathbb Y}
\newcommand{\WW}{\mathbb W}
\newcommand{\BB}{\mathbb B}

\newcommand{\UU}{\mathbb U}

\newcommand{\SSS}{\mathbb S}

\newcommand{\mc}[1]{\mathcal{#1}}
\newcommand{\wtil}[1]{\widetilde{#1}}
\newcommand{\opn}[1]{\operatorname{#1}}

\newcommand{\jt}{\textstyle}
\newcommand{\ee}[1]{{(#1)}}
\newcommand{\quo}{\opn{quo} }
\newcommand{\rema}{\opn{rem}}

\begin{document}
\begin{frontmatter}

\title{A Local Construction of the Smith Normal Form of
  a Matrix Polynomial}

\thanks{The authors were supported in part by the Director,
Office of Science, Computational and Technology Research, U.S.
Department of Energy under Contract No. DE-AC02-05CH11231, and
by the National Science Foundation through grant
DMS-0955078.}

\author{Jon Wilkening}
\address{Department of Mathematics,
  University of California, Berkeley, CA 94720-3840, USA}
\ead{wilken@math.berkeley.edu}

\author{Jia Yu}
\address{Department of Mathematics,
  University of California, Berkeley, CA 94720-3840, USA}
\ead{jiay@math.berkeley.edu}


\begin{abstract}
  We present an algorithm for computing a Smith form with multipliers
  of a regular matrix polynomial over a field.  This algorithm differs
  from previous ones in that it computes a local Smith form for each
  irreducible factor in the determinant separately and then combines
  them into a global Smith form, whereas other algorithms apply a
  sequence of unimodular row and column operations to the original
  matrix.  The performance of the algorithm in exact arithmetic is
  reported for several test cases.
\end{abstract}

\begin{keyword}
Matrix polynomial, canonical forms, Smith form, Jordan chain,
symbolic computation

\vspace*{5pt} {\bf AMS subject classifications.} 68W30, 15A21,
11C20, 11C08
\end{keyword}

\end{frontmatter}

\section{Introduction}

Canonical forms are a useful tool for classifying matrices,
identifying their key properties, and reducing complicated systems of
equations to the de-coupled, scalar case.  When working with matrix
polynomials over a field $K$, one fundamental canonical form, the
Smith form, is defined. It is a diagonalization
\begin{equation} \label{eqn:smith}
  A(\lambda)=E(\lambda)D(\lambda)F(\lambda)
\end{equation}
of the given matrix $A(\lambda)$ by unimodular matrices $E(\lambda)$
and $F(\lambda)$ such that the diagonal entries $d_i(\lambda)$ of
$D(\lambda)$ are monic polynomials and $d_i(\lambda)$ is divisible by
$d_{i-1}(\lambda)$ for $i\ge2$.

This factorization has various applications.  The most common one
\cite{Goh1982,kailath80,neven93} involves solving the system of differential
equations
\begin{equation}\label{eq_ode}
A^{(q)}\frac{d^q x}{d t^q}+\dots+A^{(1)}\frac{d x}{d t}+A^{(0)}
x=f(t),
\end{equation}
where $A^{(0)}$, \dots, $A^{(q)}$ are $n\times n$ matrices over $\CC$.
For brevity, we denote this system by $A(d/dt)x=f$, where
$A(\lambda)=A^{(0)}+A^{(1)} \lambda+\dots+A^{(q)} \lambda^q$. Assume
for simplicity that $A(\lambda)$ is regular, i.e.~$\det[A(\lambda)]$
is not identically zero, and that (\ref{eqn:smith}) is a Smith form of
$A(\lambda)$. The system \eqref{eq_ode} is then equivalent to
\[\left(
    \begin{array}{ccc}
      d_1(\frac{d}{dt}) &  &  \\
       & \ddots &  \\
       &  & d_n(\frac{d}{dt}) \\
    \end{array}
  \right)\left(
           \begin{array}{c}
             y_1\\
             \vdots\\
             y_n\\
           \end{array}
         \right)=\left(
           \begin{array}{c}
             g_1\\
             \vdots\\
             g_n\\
           \end{array}
         \right),
\]
where $y=F(d/dt)x(t)$ and $g=E^{-1}(d/dt)f(t)$.  Note that
$E^{-1}(\lambda)$ is a matrix polynomial over $\CC$ due to the
unimodularity of $E(\lambda)$. This system splits into $n$ independent
scalar ordinary differential equations
\[d_i\Big(\frac{d}{dt}\Big)y_i(t)=g_i(t),\qquad 1\le i\le n,\]
and the solution of \eqref{eq_ode} is then given by $x=F^{-1}(d/dt)
y$, where $F^{-1}(\lambda)$ is also a matrix polynomial over $\CC$.

Another important application of the Smith form concerns the study of
the algebraic structural properties of systems in linear control
theory \cite{kailath80,rosenbrock70,vanDooren83}.  For example, the
zeros and poles of a multivariable transfer function $H(s)$ are
revealed by the Smith-McMillan form of $H(s)$, which is a close
variant of the Smith form, but for rational (as opposed to polynomial)
matrices.  In many applications, one only needs to compute a minimal
basis for the kernel of a matrix polynomial.  Specialized algorithms
\cite{neven93,zuniga09} have been developed for this sub-problem of the
Smith form calculation.

Smith forms of linear matrix polynomials (i.e.~matrix pencils) are
related to the concept of similarity of matrices. A fundamental
theorem in matrix theory \cite{Gant1960,Goh1982} states that two
square matrices $A$ and $B$ over a field $K$ are similar if and only
if their characteristic matrix polynomials $\lambda I-A$ and $\lambda
I-B$ have the same Smith form $D(\lambda)$.  Other applications of
this canonical form include finding the Frobenius form
\cite{Villard94, Villard1997} of a matrix $A$ over a field by
computing the invariant factors of the matrix pencil $\lambda I-A$.

Many algorithms have been developed for the computation of canonical
forms of matrix polynomials in floating point arithmetic.  One common
approach involves finding an equivalent linear matrix pencil with the
same finite zeros as the original matrix polynomial and a closely
related Smith form \cite{vanDooren83}.  The Kronecker form
\cite{vanDooren79, vanDooren83,demmel} of the matrix pencil is then
computed to determine the eigenstructure of the original polynomial
matrix.  Another approach centers around computing the local spectral
structure of a matrix polynomial at a single complex root, $\lambda_0$,
of the characteristic determinant \cite{vanSchagen,Wilken2007}.
These methods usually boil down to computing kernels of nested
Toeplitz matrices \cite{Wilken2007,zuniga09}.  One advantage of this
local approach over the global matrix pencil approach is that only a
few terms in an expansion of the matrix polynomial in powers of
$\lambda-\lambda_0$ are needed to compute the spectral behavior. This
can lead to a significant computational savings, and also allows for
generalization from matrix polynomials to analytic matrix functions
\cite{vanSchagen,Wilken2007}.  Such local canonical forms can be used
to efficiently compute successive terms in the Laurent expansion of
the inverse of an analytic matrix \cite{howlett,Wilken2007}.  Backward
stability analysis of the effect of roundoff error may be found in
\cite{vanDooren83,Wilken2007,zuniga09}.  A geometric approach to the
perturbation theory of matrix pencils is discussed in \cite{edelman1}.

The symbolic computation of Smith forms of matrices over
$\QQ[\lambda]$ is also a widely studied topic. Kannan \cite{Kann1985}
gave a method for computing the Smith form with repeated
triangularizations of the matrix polynomial over $\QQ$. Kaltofen,
Krishnamoorthy and Saunders \cite{Kal1987} gave the first polynomial
time algorithm for the Smith form (without multipliers) using the
Chinese remainder theorem.  A new class of probabilistic algorithms
(the Monte Carlo algorithms) were proposed by Kaltofen, Krishnamoorthy
and Saunders
%
%
\cite{Kal1987, Kal1990}. They showed that by multiplying the
given matrix polynomial by a randomly generated constant matrix on
the right, the Smith form with multipliers is obtained with high
probability by two steps of computation of the Hermite form. A Las
Vegas algorithm given by Storjohann and Labahn \cite{Stor1994,
Stor1997} significantly improved the complexity by rapidly checking
the correctness of the result of the KKS algorithm. Villard
\cite{Villard93, Villard95} established the first deterministic
polynomial-time method to obtain the Smith form with multipliers by
explicitly computing a good-conditioning matrix that replaces the
random constant matrix in the Las Vegas algorithm. Villard also applied
the method of Marlin, Labhalla and Lombardi \cite{lab1996} to obtain
useful complexity bounds for the algorithm.

We propose a new deterministic algorithm for the symbolic computation
of Smith forms of matrix polynomials over a field in
Section~\ref{sec_algorithm}.  Our approach differs from previous
methods in that we begin by constructing local diagonal forms that we
later combine to obtain a (global) post-multiplier.
Although we do not discuss complexity bounds, we compare the
performance of our algorithm to Villard's method with good
conditioning in Section~\ref{sec_numerical}, and discuss the reasons
for the increase in speed.  The new algorithm is also easy to
parallelize.  In Appendix~\ref{appendix}, we present an algebraic
framework that connects this work to \cite{Wilken2007}, and give a
variant of the algorithm in which all operations are done in the field
$K$ rather than manipulating polynomials directly.

As mentioned above, local canonical forms have been used successfully
to study the structure of a matrix polynomial near a single root
$\lambda_0\in\CC$ of the characteristic determinant.  An important
point that has been neglected in the literature is that these roots
$\lambda_0$ may not be expressible in radicals, or may involve such
complicated expressions that current algorithms can only be carried
out in floating point arithmetic.  A major goal of this paper is to
develop a machinery for computing local forms for all the complex
roots of a $\QQ$-irreducible factor $p(\lambda)$ of the characteristic
determinant simultaneously, without having to resort to floating point
arithmetic at each root separately.  This is done by working over the
fields $\QQ$ or $\QQ+i\QQ$ rather than $\RR$ or $\CC$ when computing
local forms.


\section{Preliminaries}
In this section, we describe the theory of Smith forms of matrix
polynomials over a field $K$, which follows the definition in
\cite{Goh1982} over $\CC$. In practice, $K$ will be
$\QQ$, $\QQ+i\QQ$, $\RR$, or $\CC$, but it is convenient to deal
with all these cases simultaneously. We also give a brief review of the theory
of Jordan chains as well as B\'{e}zout's identity, which play an
important role in our algorithm for computing Smith forms of matrix
polynomials.

\subsection{Smith Forms}
Suppose $A(\lambda)=\sum_{k=0}^q A^{(k)} \lambda^k$ is an $n\times n$
matrix polynomial, where $A^{(k)}$ are $n\times n$ matrices whose
entries are in a field $K$.
Assuming that $A(\lambda)$ is \emph{regular}, i.e.~the determinant of
$A(\lambda)$ is not identically zero, the following theorem is proved
in \cite{Goh1982} (for $K=\CC$).

\begin{theorem}\label{smith} There exist matrix polynomials $E(\lambda)$
  and $F(\lambda)$ over $K$ of size $n\times n$, with constant nonzero
  determinants, such that
\begin{equation}\label{eq_smith}
  A(\lambda)=E(\lambda)D(\lambda)F(\lambda), \qquad
  D(\lambda) = \opn{diag}[d_1(\lambda),\dots,d_n(\lambda)],
\end{equation}
where $D(\lambda)$ is a diagonal matrix with monic scalar polynomials
$d_i(\lambda)$ over $K$ such that $d_i(\lambda)$ is divisible by
$d_{i-1}(\lambda)$.
\end{theorem}

Since $E(\lambda)$ and $F(\lambda)$ have constant nonzero
determinants, \eqref{eq_smith} is equivalent to
\begin{equation}\label{eq_smith_2}
U(\lambda)A(\lambda)V(\lambda)=D(\lambda),
\end{equation}
where
$U(\lambda):=E(\lambda)^{-1}$ and $V(\lambda):=F(\lambda)^{-1}$
are also matrix polynomials over $K$.

\begin{definition} The representation in \eqref{eq_smith} or
\eqref{eq_smith_2}, or often $D(\lambda)$
alone, is called a {\em
Smith form} of $A(\lambda)$. The matrices $U(\lambda)$, $V(\lambda)$
are known as {\em multipliers}.
Square matrix polynomials with
constant nonzero determinants like $E(\lambda)$ and $F(\lambda)$
are called {\em unimodular}.
\end{definition}

The diagonal matrix $D(\lambda)$ in the Smith form is unique, while
the representation \eqref{eq_smith} is not. Suppose that
\begin{equation}
  \Delta(\lambda):=\det[A(\lambda)]
\end{equation}
can be decomposed into prime elements $p_1(\lambda), \dots,
p_l(\lambda)$ in the principal ideal domain $K[\lambda]$, that is,
$\Delta(\lambda)=c\prod_{j=1}^l p_j(\lambda)^{\kappa_j}$ where
$c\neq0$ is in the field $K$, $p_j(\lambda)$ is monic and irreducible, and $\kappa_j$ are positive integers for $j=1, \dots, l$. Then the
$d_i(\lambda)$ are given by
\[d_i(\lambda)=\prod_{j=1}^l p_j(\lambda)^{\kappa_{ji}},\qquad
(1\le i\le n)
\]
for some integers $0\leq\kappa_{j1}\leq\dots\leq\kappa_{jn}$
satisfying $\sum_{i=1}^n\kappa_{ji}=\kappa_j$ for $j=1,\dots,l$.

We now define a \emph{local Smith form} for $A(\lambda)$ at
$p(\lambda)$.  Let $p(\lambda)=p_j(\lambda)$ be one of the irreducible
factors of $\Delta(\lambda)$ and define $\alpha_i=\kappa_{ji}$,
$\mu=\kappa_j$.  Generalizing the case that
$p(\lambda)=\lambda-\lambda_j$, we call $\mu$ the algebraic
multiplicity of $p(\lambda)$.

\begin{theorem}\label{local_smith}
Suppose $A(\lambda)$ is an $n\times n$ matrix over $K[\lambda]$ and
$p(\lambda)$ is an irreducible factor of $\Delta(\lambda)$.  There
exist $n\times n$ matrix polynomials $E(\lambda)$ and $F(\lambda)$
such that
\begin{equation}\label{eq_local_smith1}
  A(\lambda)=E(\lambda)D(\lambda)F(\lambda), \qquad
  D(\lambda)=\opn{diag}[p(\lambda)^{\alpha_1},\dots,p(\lambda)^{\alpha_n}],
\end{equation}
where $0\leq\alpha_1\leq\dots\leq\alpha_n$ are non-negative integers
and $p(\lambda)$ does not divide $\det[E(\lambda)]$ or $\det[F(\lambda)]$.
\end{theorem}

$E(\lambda)$ and $F(\lambda)$ are not uniquely determined in a local
Smith form. In particular, we can impose the additional requirement that $F(\lambda)$ be
unimodular by absorbing the missing parts of $D(\lambda)$ in Theorem
\ref{smith} into $E(\lambda)$. Then the local Smith form of
$A(\lambda)$ at $p(\lambda)$ is given by
\begin{equation}\label{eq_local_smith2}
A(\lambda)V(\lambda)=E(\lambda)D(\lambda),
\end{equation}
where $V(\lambda):=F(\lambda)^{-1}$ is a matrix polynomial.

\subsection{Multiplication and division in $R/pR$}
\label{sec:mult:RpR}
We define $R=K[\lambda]$ and $M=R^n$.  Note that $R$ is a principal
ideal domain and $M$ is a free $R$-module of rank $n$.  Suppose $p$ is
a prime element in $R$.  Since $p$ is irreducible, $R/pR$ is a field
and $M/pM$ is a vector space over this field.

Multiplication and division in $R/pR$ are easily carried
out using the companion matrix of $p$.  If we set $s:=\deg p$
and define $\gamma:K^s\rightarrow R/pR$ by
\begin{equation}
  \gamma(x)(\lambda) =
  x^\ee0 + \cdots + \lambda^{s-1}x^\ee{s-1} + pR,
  \qquad x = \big(x^\ee0;\dots;x^\ee{s-1}\big)\in K^s,
\end{equation}
we can pull back the field structure of $R/pR$ to $K^s$ to obtain
\begin{equation}
\label{eq_xyRpR}
\begin{aligned}
  xy=\gamma(x)(S)y&=[x^{(0)}I+x^{(1)}S+\cdots+x^{(s-1)}S^{s-1}]y \\
&= [y,Sy,\dots,S^{s-1}y]x = [x,Sx,\dots,S^{s-1}x]y
\end{aligned}
\end{equation}
and $x/y = [y,Sy,\dots,S^{s-1}y]^{-1}x$, where
\begin{equation} \label{eq_S_def}
  S=\left(
      \begin{array}{cccc}
        0 & \ldots & 0 & -a_0 \\[-1pt]
        1 & \ddots & \vdots & \vdots \\[-5pt]
         & \ddots & 0 & -a_{s-2} \\
        0 &  & 1 & -a_{s-1}
      \end{array}
    \right)
\end{equation}
is the companion matrix of $p(\lambda) = a_0 + a_1\lambda + \cdots +
a_{s-1}\lambda^{s-1} + \lambda^s$.  Note that $S$ represents
multiplication by $\lambda$ in $R/pR$.  The matrix
$[y,Sy,\dots,S^{s-1}y]$ is invertible when $y\ne0$ since a non-trivial
vector $x$ in its kernel would lead to non-zero polynomials
$\gamma(x),\gamma(y)\in R/pR$ whose product is zero (mod $p$), which
is impossible as $p$ is irreducible.

\subsection{Jordan Chains}\label{sec_JC}
Finding a local Smith form of a matrix polynomial over $\CC$ at
$p(\lambda)=\lambda-\lambda_0$ is equivalent to finding a canonical
system of Jordan chains \cite{vanSchagen,Wilken2007} for $A(\lambda)$ at
$\lambda_0$.  We now generalize the notion of Jordan chain to the case
of an irreducible polynomial over a field $K$.

\begin{definition}
  Suppose $A(\lambda)$ is an $n\times n$ matrix polynomial over
  a field $K$, $p(\lambda)$ is irreducible in $K[\lambda]$, and
  $\alpha\ge1$ is an integer.
  A vector polynomial $x(\lambda)\in M=K[\lambda]^n$ is called a
  {\em root function of order $\alpha$} for $A(\lambda)$ at
  $p(\lambda)$ if
\begin{equation}\label{eq_J_chain}
  A(\lambda)x(\lambda)=O(p(\lambda)^\alpha)
\end{equation}
and $p(\lambda)\nmid x(\lambda)$.  The meaning of (\ref{eq_J_chain})
is that each component of $A(\lambda)x(\lambda)$ is divisible by
$p(\lambda)^\alpha$.  If the root function $x(\lambda)$ has the form
\begin{equation}\label{eqn:x:expand}
  x(\lambda)=x^{(0)}(\lambda) + p(\lambda)x^\ee1(\lambda) + \cdots
  + p(\lambda)^{\alpha-1}x^\ee{\alpha-1}(\lambda)
\end{equation}
with $\deg x^{(k)}(\lambda)<s:=\deg p(\lambda)$, the coefficients
$x^\ee{k}(\lambda)$ are said to form a {\em Jordan chain of length
  $\alpha$} for $A(\lambda)$ at $p(\lambda)$.  A root function can
always be converted to the form (\ref{eqn:x:expand}) by truncating or
zero-padding its expansion in powers of $p(\lambda)$.  If $K$ can be
embedded in $\CC$, (\ref{eq_J_chain}) implies that over $\CC$,
$x(\lambda)$ is a root function of $A(\lambda)$ of order $\alpha$ at
each root $\lambda_j$ of $p(\lambda)$ simultaneously.
\end{definition}

\begin{definition} \label{def:canon} Several vector polynomials
  $\{x_j(\lambda)\}_{j=1}^\nu$ form a {\em system of root functions}
  at $p(\lambda)$ if
\begin{equation} \label{sys_root}
  \begin{array}{l}
    1.\,\, A(\lambda)x_j(\lambda)=O(p(\lambda)^{\alpha_j}), \qquad
    (\alpha_j\ge1,\quad 1\leq j\leq \nu) \\[3pt]
    \parbox{.8\linewidth}{$2.$ The set $\{\dot{x}_j(\lambda)\}_{j=1}^\nu$
      is linearly independent in $M/pM$ over $R/pR$, \\
      \hspace*{1.3in} where $R=K[\lambda]$, $M=R^n$,
      $\dot{x}_j=x_j+pM$.}
  \end{array}
\end{equation}
It is called {\em canonical} if (1) $\nu=\dim \ker\dot{A}$, where
$\dot{A}$ is the linear operator on $M/pM$ induced by $A(\lambda)$; (2)
$x_1(\lambda)$ is a root function of maximal order $\alpha_1$; and (3)
for $i>1$, $x_i(\lambda)$ has maximal order $\alpha_i$ among all root
functions $x(\lambda)\in M$ such that $\dot{x}$ is linearly
independent of $\dot{x}_1,\dots, \dot{x}_{i-1}$ in $M/pM$.  The
integers $\alpha_1\geq\dots\geq\alpha_\nu$ are uniquely
determined by $A(\lambda)$. We call $\nu$ the \emph{geometric multiplicity} of $p(\lambda)$.
\end{definition}

\begin{definition} An \emph{extended system of root functions}
  $x_1(\lambda)$,\dots,$x_n(\lambda)$ is a collection of vector
  polynomials satisfying (\ref{sys_root}) with $\nu$ replaced by $n$
  and $\alpha_j$ allowed to be zero. The extended system is said to be
  \emph{canonical} if, as before, the orders $\alpha_j$ are chosen to be
  maximal among root functions not in the span of previous root
  functions in $M/pM$.  The resulting sequence of numbers
  $\alpha_1\ge\cdots\ge\alpha_\nu\ge \alpha_{\nu+1}=\cdots=\alpha_n=0$
  is uniquely determined by $A(\lambda)$.\end{definition}

Given such a system (not necessarily canonical), we define the matrices
\begin{align}
  V(\lambda) &= [x_1(\lambda),\dots,x_n(\lambda)], \\
  D(\lambda) &= \opn{diag}[p(\lambda)^{\alpha_1},\dots,
  p(\lambda)^{\alpha_n}], \\
  E(\lambda) &= A(\lambda)V(\lambda)D(\lambda)^{-1}.
\end{align}
$E(\lambda)$ is a polynomial since column $j$ of $A(\lambda)V(\lambda)$
is divisible by $p(\lambda)^{\alpha_j}$.
The following theorem shows that aside from a reversal of the
convention for ordering the $\alpha_j$, finding a local Smith form is
equivalent to finding an extended canonical system of root functions:
\begin{theorem} \label{thm:equiv}
  The following three conditions are equivalent:

\vspace*{1pt}
\begin{center}
 \parbox{.9\linewidth}{
  \begin{itemize}
  \item[(1)]
    the columns $x_j(\lambda)$ of $V(\lambda)$ form an extended
    canonical system of root functions for $A(\lambda)$ at
    $p(\lambda)$ (up to a permutation of columns).
  \item[(2)] $p(\lambda)\nmid\det[E(\lambda)]$.
  \item[(3)] $\sum_{j=1}^n\alpha_j=\mu$, where $\mu$ is the algebraic
    multiplicity of $p(\lambda)$ in $\Delta(\lambda)$.
  \end{itemize}
}
\end{center}
\end{theorem}

This theorem is proved e.g.~in \cite{vanSchagen} for the case that
$K=\CC$.  The proof over a general field $K$ is identical, except that
the following lemma is used in place of invertibility of
$E(\lambda_0)$.  This lemma also plays a fundamental role in our
construction of Jordan chains and local Smith forms.
\begin{lemma} \label{lem:det} Suppose $K$ is a field, $p$ is an
  irreducible polynomial in $R=K[\lambda]$, and $E=[y_1,\dots,y_n]$ is
  an $n\times n$ matrix with columns $y_j\in M=R^n$.  Then
  $p\nmid\det(E)\,\Leftrightarrow \{\dot{y}_1,\dots,\dot{y}_n\}$ are
  linearly independent in $M/pM$ over $R/pR$.
\end{lemma}

\begin{proof} The $\dot{y}_j$ are linearly independent iff the
  determinant of $\dot E$ (considered as an $n\times n$ matrix with
  entries in the field $R/pR$) is non-zero.  But
  \begin{equation}
    \det\dot E = \det E + pR,
  \end{equation}
  where $\det E$ is computed over $R$.  The result follows.
\end{proof}

\subsection{B\'{e}zout's Identity}\label{sec_bezout}
As $K[\lambda]$ is a principal ideal domain, B\'{e}zout's Identity
holds, which is our main tool for combining local Smith forms into a
single global Smith form.  We define the notation $\gcd(f_1,\dots,
f_l)$ to be $0$ if each $f_j$ is zero, and the monic greatest common
divisor (GCD) of $f_1,\dots, f_l$ over $K[\lambda]$, otherwise.

\begin{theorem} {\em (B\'{e}zout's Identity)} For any
two polynomials $f_1$ and $f_2$ in $K[\lambda]$, where
$K$ is a field, there exist polynomials $g_1$ and $g_2$ in
$K[\lambda]$ such that
\begin{equation}
g_1 f_1+g_2 f_2= \gcd(f_1,f_2).
\end{equation}
\end{theorem}

B\'{e}zout's Identity can be extended to combinations of more than
two polynomials:

\begin{theorem} {\em (Generalized B\'{e}zout's Identity)} For any
scalar polynomials $f_1,\dots,f_l$ in $K[\lambda]$, there
exist polynomials $g_1,\dots,g_l$ in $K[\lambda]$ such
that
\[\sum_{j=1}^l g_j f_j= \gcd(f_1,\dots,f_l).
\]
The polynomials $g_j$ are called the {\em B\'{e}zout coefficients}
of $\{f_1,\dots,f_l\}$.
\end{theorem}

In particular, suppose we have $l$ distinct prime elements
$\{p_1,\dots,p_l\}$ in $K[\lambda]$, and $f_j$ is given by
$f_j=\prod_{k\neq j}^l p_k^{\beta_k}$, where $\beta_1,\dots,\beta_l$
are given positive integers and the notation $\prod_{k\ne j}^l$
indicates a product over all indices $k=1,\dots,l$ except $k=j$. Then
$\gcd\left(f_1,\dots,f_l\right)=1$, and we can find $g_1,\dots,g_l$
in $K[\lambda]$ such that
\begin{equation}\label{eq_bezout}
  \sum_{j=1}^l g_j f_j= 1.
\end{equation}
In this case, the polynomials $g_j$ are uniquely
determined by requiring $\deg(g_j)<s_j\beta_j$, where $s_j=\deg(p_j)$.
The formula \eqref{eq_bezout} modulo $p_k$ shows that $g_k$ is not
divisible by $p_k$.

The B\'{e}zout coefficients are easily computed using the extended
Euclidean algorithm \cite{clr}.  In practice, we use {\sf
  MatrixPolynomialAlgebra[HermiteForm]} in Maple to find a unimodular
matrix $Q$ such that
\begin{equation}\label{eqn:hermite}
  Q
  \begin{pmatrix} f_1\\f_2\\\vdots\\f_l \end{pmatrix} =
  \begin{pmatrix} r\\0\\\vdots\\0 \end{pmatrix},
\end{equation}
where $r=\gcd(f_1,\dots,f_l)=1$. The first row of $Q$ is
$[g_1,\dots,g_l]$. One could avoid computing the remaining rows of $Q$
by storing the sequence of elementary unimodular operations required
to reduce $[f_1;\dots;f_l]$ to $[r;0;\dots;0]$ and applying them to
the row vector $[1,0,\dots,0]$ from the right to obtain
$[g_1,\dots,g_l]$.

\section{An Algorithm for Computing a (global) Smith Form}
\label{sec_algorithm}

In this section, we describe an algorithm for computing a Smith form
of a regular $n\times n$ matrix polynomial $A(\lambda)$ over a field
$K$.  We have in mind the case where $K=\CC$, $\RR$, $\QQ$ or
$\QQ+i\QQ\subset\CC$, but the construction works for any field.
The basic procedure follows several steps, which will be explained
further below:
\begin{itemize}
\item Step 0. Compute $\Delta(\lambda)=\det[A(\lambda)]$ and
decompose it into irreducible monic factors in $K[\lambda]$,
\begin{equation}
  \Delta(\lambda)=\opn{const}\cdot p_1(\lambda)^{\kappa_1}\dots
  p_l(\lambda)^{\kappa_l}.
\end{equation}
\item Step 1. Compute a local Smith form
  \begin{equation}\label{eq_local_smith_j}
    A(\lambda)V_j(\lambda)=E_j(\lambda)
    \opn{diag}\big[p_j(\lambda)^{\kappa_{j1}},\dots,p_j(\lambda)^{\kappa_{jn}}\big]
  \end{equation}
  for each factor $p_j(\lambda)$ of $\Delta(\lambda)$.

\item Step 2. Find a linear combination $B_n(\lambda)=\sum_{j=1}^l
  g_j(\lambda)f_j(\lambda) V_j(\lambda)$ using the B\'{e}zout
  coefficients of $f_j(\lambda)=\prod_{k\ne
    j}^lp_k(\lambda)^{\kappa_{kn}}$ so that the columns of
  $B_n(\lambda)$ form an extended canonical system of root functions
  for $A(\lambda)$ with respect to each $p_j(\lambda)$.

\item Step 3. Eliminate extraneous zeros from $\det\big[A(\lambda)
  B_n(\lambda)\big]$ by finding a unimodular matrix $V(\lambda)$ such
  that $B_1(\lambda)=V(\lambda)^{-1}B_n(\lambda)$ is lower triangular.
  We will show that $A(\lambda)V(\lambda)$ is then of the form
  $E(\lambda)D(\lambda)$ with $E(\lambda)$ unimodular and $D(\lambda)$
  as in (\ref{eq_smith}).

\end{itemize}

\begin{rmk}
Once the local Smith forms are known,
the diagonal entries of the matrix polynomial $D(\lambda)$ are
given by
\[d_i(\lambda)=\prod_{j=1}^lp_j(\lambda)^{\kappa_{ji}}, \quad
i=1,\dots,n.
\]
This allows us to order the
columns once and for all in Step 2.
\end{rmk}

\subsection{A Local Smith Form Algorithm (Step 1)}\label{sec_step1}
In this section, we show how to generalize the construction in
\cite{Wilken2007} for finding a canonical system of Jordan chains for
an analytic matrix function $A(\lambda)$ over $\CC$ at $\lambda_0=0$
to finding a local Smith form for a matrix polynomial $A(\lambda)$
with respect to an irreducible factor $p(\lambda)$ of
$\Delta(\lambda)=\det[A(\lambda)]$.  The new algorithm reduces to the
``exact arithmetic'' version of the previous algorithm when
$p(\lambda)=\lambda$.  In Appendix~\ref{appendix}, we present a
variant of the algorithm that is easier to implement than the current
approach, and is closer in spirit to the construction in
\cite{Wilken2007}, but is less efficient by a factor of $s=\deg p$.

Our goal is to find matrices $V(\lambda)$ and $E(\lambda)$ such
that $p(\lambda)$ does not divide $\det[V(\lambda)]$ or
$\det[E(\lambda)]$, and such that
\begin{equation}\label{eq_loc}
  A(\lambda)V(\lambda) = E(\lambda)D(\lambda), \qquad
  D(\lambda) = \opn{diag}[p(\lambda)^{\alpha_1},\dots,p(\lambda)^{
    \alpha_n}],
\end{equation}
where $0\le\alpha_1\le\cdots\le\alpha_n$.  In our construction,
$V(\lambda)$ will be unimodular, which reduces the work in Step 3 of
the high level algorithm, the step in which extraneous zeros are
removed from the determinant of the combined local Smith forms.

We start with $V(\lambda)=I_{n\times n}$ and perform a sequence of
column operations on $V(\lambda)$ that preserve its determinant (up to
a sign) and systematically increase the orders $\alpha_i$ in
$D(\lambda)$ in (\ref{eq_loc}) until $\det[E(\lambda)]$ no longer
contains a factor of $p(\lambda)$.  This can be considered a ``breadth
first'' construction of a canonical system of Jordan chains, in
contrast to the ``depth first'' procedure described in
Definition~\ref{def:canon} above.

\begin{figure}[t]
\begin{center}
\fbox{\parbox{.9\linewidth}{
\begin{algori} \label{alg1} (Local Smith form, preliminary version)
\vspace*{4pt}
\begin{tabbing}
\hspace*{.125in} \= \hspace*{.125in} \= \hspace*{.125in} \=
\hspace*{.02in} \= \hspace*{.105in} \=
\hspace*{1.75in} \= \kill
\> $k=0$,\; $i=1$,\; $V=[x_1,\dots,x_n]=I_{n\times n}$ \\[0pt]
\> \textbf{while} $i\le n$ \\[0pt]
\>\> $r_{k-1} = n+1-i$ \qquad\qquad \emph{$r_{k-1}:=$ dim.~of space
of
J. chains of length $\ge k$} \\[0pt]
\>\> \textbf{for} $j=1,\dots,r_{k-1}$ \\[0pt]
\>\>\> $y_i = \rema(\quo(Ax_i,p^k), p)$
\>\>\> \qquad\emph{define $y_i$ so $Ax_i=p^k y_i+O(p^{k+1})$} \\[0pt]
\>\>\> \textbf{if} the set $\{\dot{y}_1,\dots,\dot{y}_i\}$ is
linearly
independent in $M/pM$ over $R/pR$\\[0pt]
\>\>\>\>\> $\alpha_i = k$, \; $i=i+1$ \>
\qquad\emph{accept $x_i$ and $y_i$, define $\alpha_i$} \\[0pt]
\>\>\> \textbf{else} \\[0pt]
\>\>\>\>\> find $\dot{a}_1$, \dots, $\dot{a}_{i-1}\in R/pR$ so
that
$\dot{y}_i - \sum_{m=1}^{i-1} \dot{a}_m \dot{y}_m=\dot{0}$ \\[3pt]
\>\>\>\> $\star$ \> $x_i^{(new)} = x_i^{(old)} - \sum_{m=1}^{i-1}
p^{k-\alpha_m} a_m x_m$ \\[1pt]
\>\>\>\>\> $x_\text{tmp} = x_i$, \; $x_m=x_{m+1}, \;
(m=i,\dots,n-1)$, \;
$x_n = x_\text{tmp}$ \\[0pt]
\>\>\> \textbf{end} if \\[0pt]
\>\> \textbf{end} for $j$ \\[0pt]
\>\> $k=k+1$ \\[0pt]
\> \textbf{end} while \\[0pt]
\> $\beta=k-1$, \; $r_\beta = 0$ \>\>\>\>\>
\emph{$\beta:=\alpha_n=$ maximal Jordan chain length}
\end{tabbing}
\end{algori}
}}
\end{center}

\caption{Algorithm for computing a local Smith form.  Here
$\quo(\cdot,\cdot)$ and $\rema(\cdot,\cdot)$ are the quotient and
remainder of polynomials:
$g=\quo(f, p)$, $r=\rema(f, p)$ $\Leftrightarrow$
$f=gp+r$, $\deg r<\deg p$. }
\label{fig:alg1}
\end{figure}

The basic algorithm is presented in Figure~\ref{fig:alg1}. The idea is
to run through the columns of $V$ in turn and ``accept'' columns
whenever the leading term of the residual $A(\lambda)x_i(\lambda)$ is
linearly independent of its predecessors; otherwise we find a linear
combination of previously accepted columns to cancel this leading term
and cyclically rotate the column to the end for further processing.
Note that for each $k$, we cycle through each unaccepted column
exactly once: after rotating a column to the end, it will not become
active again until $k$ has increased by one.  At the start of the
\emph{while} loop, we have the invariants

\vspace*{2pt}
\begin{tabbing}
\hspace*{.25in} \= \hspace*{2in} \= \kill
\> (1) $Ax_m$ is divisible by $p^k$, \> $(i\le m\le n).$ \\
\> (2) $Ax_m = p^{\alpha_m} y_m + O(p^{\alpha_m+1})$, \>
  $(1\le m < i).$ \\
\> (3) if $i\ge2$ then $\{\dot{y}_m\}_{m=1}^{i-1}$ is linearly
  independent in $M/pM$ over $R/pR$.
\end{tabbing}
\vspace*{2pt}

\noindent
The third property is guaranteed by the \emph{if} statement, and the
second property follows from the first due to the definition of
$\alpha_i$ and $y_i$ in the algorithm.  The first property is
obviously true when $k=0$; it continues to hold each time $k$ is
incremented due to step~$\star$, after which $Ax_i^{(new)}$ is
divisible by $p^{k+1}$:
\begin{align*}
  Ax_i^{(old)} - \sum_{m=1}^{i-1}p^{k-\alpha_m}a_m Ax_m &=
  p^ky_i + O(p^{k+1}) - \sum_{m=1}^{i-1}p^{k-\alpha_m}a_m\Big(
  p^{\alpha_m}y_m + O(p^{\alpha_m+1})\Big) \\ &=
  p^k\Big(y_i - \sum_{m=1}^{i-1}a_my_m\Big) + O(p^{k+1}) =
  O(p^{k+1}).
\end{align*}
This equation is independent of which polynomials $a_m\in R$ are
chosen to represent $\dot{a}_m\in R/pR$, but different choices will
lead to different (equally valid) Smith forms; in practice, we choose
the unique representatives such that $\deg a_m<s$, where
\begin{equation}
  s = \deg p.
\end{equation}
This choice of the $a_m$ leads to two additional invariants of the
\emph{while} loop, namely
\begin{tabbing}
\hspace*{.25in} \= \hspace*{2in} \= \kill
\> (4) $\deg x_m \le \max(s\alpha_m-1,0)$, \> $(1\le m<i),$ \\
\> (5) $\deg x_m \le \max(sk-1,0)$, \> $(i\le m\le n),$
\end{tabbing}
which are easily proved inductively by noting that
\begin{equation}
  \deg(p^{k-\alpha_m}a_mx_m)\le s(k-\alpha_m)+(s-1)+\deg(x_m)\le s(k+1)-1.
\end{equation}
The \emph{while} loop eventually terminates, for at the end of
each loop (after $k$ has been incremented) we have produced a
unimodular matrix $V(\lambda)$ such that
\begin{equation}
  A(\lambda)V(\lambda) = E(\lambda)D(\lambda), \qquad
  D=\opn{diag}[p^{\alpha_1},\dots,p^{\alpha_{i-1}},\underbrace{
    p^k,\dots,p^k}_{r_{k-1}\text{ times}}].
\end{equation}
Hence, the algorithm must terminate before $k$ exceeds the
algebraic multiplicity $\mu$ of $p(\lambda)$ in $\Delta(\lambda)$:
\begin{equation}
\jt  k \le \big(\sum_{m=1}^{i-1}\alpha_i\big) + (n+1-i)k \le \mu, \qquad\quad
  \Delta(\lambda)=f(\lambda)p(\lambda)^\mu, \quad p\nmid f.
\end{equation}
In fact, we can avoid the last iteration of the \emph{while}
loop if we change the test to
\begin{equation*}
  \jt\textbf{while }
  \big[\big(\sum_{m=1}^{i-1}\alpha_i\big) + (n+1-i)k\big]<\mu
\end{equation*}
and change the last line to
\begin{equation*}
  \beta=k, \qquad
  \alpha_m=k, \quad (i\le m\le n), \qquad
  r_{\beta-1}=n+1-i, \qquad
  r_\beta = 0.
\end{equation*}
We know the remaining columns of $V$ will be
accepted without having to compute the remaining $y_i$ or
check them for linear independence.  When the algorithm
terminates, we will have found a unimodular matrix $V(\lambda)$
satisfying (\ref{eq_loc}) such that the columns of
\begin{equation*}
  \dot{E}(\lambda)=[\dot{y}_1(\lambda),\dots,\dot{y}_n(\lambda)]
\end{equation*}
are linearly independent in $M/pM$ over $R/pR$.  By Lemma~\ref{lem:det},
$p(\lambda)\nmid \det[E(\lambda)]$, as required.

To implement the algorithm, we must find an efficient way to compute $y_i$,
test for linear independence in $M/pM$, find the coefficients $a_m$ to
cancel the leading term of the residual, and update $x_i$.  Motivated
by the construction in \cite{Wilken2007}, we interpret the loop over
$j$ in Algorithm~\ref{alg1} as a single nullspace calculation.

To this end,
we define $R_l=\{a\in R\;:\;\deg a<l\}$ and $M_l=R_l^n$, both
viewed as vector spaces over $K$.  Then we have an isomorphism
$\Lambda$ of vector spaces over $K$
\begin{equation}
\label{eq_Lam}
\begin{aligned}
  \Lambda:(M_s)^{k}&\rightarrow M_{sk}, \\ 
  \Lambda(x^\ee0;\dots;x^\ee{k-1})&=x^\ee0+px^\ee1+\cdots+p^{k-1}x^\ee{k-1}.
\end{aligned}
\end{equation}
At times it will be convenient to identify $R_{ls}$ with $R/p^lR$
and $M_{ls}$ with $M/p^lM$ to obtain ring and module structures for
these spaces.
We also expand
\begin{equation} \label{eq_Aexpand}
  A = A^\ee0 + pA^\ee1+\cdots + p^qA^\ee q,
\end{equation}
where $A^\ee j$ is an $n\times n$ matrix with entries in $R_s$.

By invariants (4) and (5) of the \emph{while} loop
in Algorithm~\ref{alg1}, we may write
$x_i=\Lambda(x_i^\ee0;\dots;x_i^\ee{\alpha})$ with
$\alpha=\max(k-1,0)$. Since $Ax_i$ is divisible by $p^k$, we have
\begin{equation} \label{eq_yi_def}
  y_i = \rema(\quo(Ax_i, p^k), p) =
  \sum_{j=0}^{k}\rema(A^\ee{k-j} x_i^\ee{j}, p)
  + \sum_{j=0}^{k-1} \quo( A^\ee{k-1-j} x_i^\ee{j}, p).
\end{equation}
The matrix-vector multiplications $A^\ee{k-j}x_i^\ee j$ are done in
the ring $R$ (leading to vector polynomials of degree $\le 2s-2$)
before the quotient and remainder are taken.  When $k=0$, the second
sum should be omitted, and when $k\ge1$, the $j=k$ term in the first
sum can be dropped since $x_i^\ee{k}=0$ in the algorithm.  
It is convenient to write (\ref{eq_yi_def}) in matrix form.
If $k=0$ we have
\begin{equation}
  [y_1,\dots,y_n] = A^\ee0.
\end{equation}
If $k\ge1$, suppose we have already computed the $nk\times r_{k-1}$
matrix $X_{k-1}$ with columns
\begin{equation}
  X_{k-1}(\,:\,,m+1-i) =
  \big(x_m^\ee0;\dots;x_m^\ee{k-1}\big), \qquad i\le m\le n.
\end{equation}
Note that
$\Lambda(X_{k-1})$ (acting column by column) contains the last
$r_{k-1}$ columns of $V(\lambda)$ at the start of the \emph{while}
loop in Algorithm~\ref{alg1}.  Then by (\ref{eq_yi_def}),
\begin{equation} \label{eq_yi1}
  [y_i,\dots,y_n] = \rema([A^\ee k,\dots,A^\ee1]X_{k-1}, p) +
  \quo([A^\ee{k-1},\dots,A^\ee0]X_{k-1}, p).
\end{equation}
As before, the matrix multiplications are done in the ring $R$ before
the quotient and remainder are computed. The
components of each $y_m$ belong to $R_s$.

Next we define the auxiliary matrices
\begin{equation}\label{eqn:mcA1}
  \mc{A}_k = \begin{cases}
    A^\ee0, & \quad k=0, \\
    \big[\mc{A}_{k-1} \;,\; [y_i,\dots,y_n]\big], &
    \quad 1\le k\le \beta-1.
  \end{cases}
\end{equation}
and compute the reduced row-echelon form of $\dot{\mc{A}}_k$
using Gauss-Jordan elimination over the field $R/pR$.
The reduced row-echelon form of $\dot{\mc{A}}_k$ can be interpreted as a
tableau telling which columns of $\dot{\mc{A}}_k$ are linearly independent
of their predecessors (the accepted columns), and also giving the
linear combination of previously accepted columns that will annihilate
a linearly dependent column.  On the first iteration (with $k=0$),
step $\star$ in Algorithm~\ref{alg1} will build up the matrix
\begin{equation}
  X_0=\opn{null}(\dot{\mc{A}}_0),
\end{equation}
where $\opn{null}(\cdot)$ is the standard algorithm for computing a
basis for the nullspace of a matrix from the reduced row-echelon form
(followed by a truncation to replace elements in $R/pR$ with their
representatives in $R_s$).  But rather than rotating these columns to
the end as in Algorithm~\ref{alg1}, we now \emph{append} the
corresponding $y_i$ to the end of $\mc{A}_{k-1}$ to form $\mc{A}_{k}$
for $k\ge1$.  The ``dead'' columns left behind (not accepted, not
active) serve only as placeholders, causing the resulting matrices
$\mc{A}_k$ to be nested.  We use $\opn{rref}(\cdot)$ to denote the
reduced row-echelon form of a matrix polynomial. The leading columns
of $\opn{rref}(\dot{\mc{A}}_{k})$ will then coincide with
$\opn{rref}(\dot{\mc{A}}_{k-1})$, and the nullspace matrices will also
be nested.  We denote the new columns of $\opn{null}(\dot{\mc{A}}_k)$
beyond those of $\opn{null}(\dot{\mc{A}}_{k-1})$ by $[Y_k; U_k]$:
\begin{equation}\label{eqn:ykuk1}
  \begin{pmatrix}X_0 & Y_1 & \cdots & Y_{k-1} & Y_k \\
    0 & [U_1;0] & \cdots & [U_{k-1};0] & U_k \end{pmatrix} :=
  \opn{null}(\dot{\mc{A}}_k).
\end{equation}
Note that $\mc{A}_k$ is $n\times(n+R_{k-1})$, where
\begin{equation}
  R_{-1} = 0, \qquad
  R_k = r_0 + \cdots + r_k = \dim\ker\dot{\mc{A}}_k, \qquad (k\ge0).
\end{equation}
We also see that $X_0$ is $n\times r_0$, $Y_k$ is $n\times r_k$,
$U_k$ is $r_{k-1}\times r_k$, and
\begin{equation}
  r_k\le r_{k-1}, \qquad (k\ge0).
\end{equation}
This inequality is due to the fact that the dimension of the kernel
cannot increase by more than the number of columns added.

If column $i$ of $\dot{\mc{A}}_k$ is linearly dependent on its
predecessors, the coefficients $a_m$ used in step $\star$ of
Algorithm~\ref{alg1} are precisely the (truncations of the)
coefficients that appear in column $i$ of
$\opn{rref}(\dot{\mc{A}}_k)$.  As shown in Figure~\ref{fig:extract},
the corresponding null vector
(i.e.~column of $[Y_k;U_k]$) contains the negatives of these
coefficients in the rows corresponding to the previously accepted
columns of $\dot{\mc{A}}_k$, followed by a 1 in row $i$.
Thus, in step $\star$, if $k\ge1$ and we
write $x_m=\Lambda\big( x_m^\ee{0};\dots;x_m^\ee{\alpha}\big)$ with
$\alpha=\max(\alpha_m-1,0)$, the update
\begin{align*}
  x_i^\ee{new} = x_i^\ee{old} - \sum_{m=1}^{i-1}
  p^{k-\alpha_m}a_mx_m, \qquad
  a_mx_m=\Lambda\big(z^\ee0;\dots;z^\ee{\alpha_m}\big), \\
  z^\ee{j} = \begin{cases}
    \rema(a_mx_m^\ee0, p), &\; j=0, \\
    \rema(a_mx_m^\ee{j}, p) +
    \quo(a_mx_m^\ee{j-1}, p), &\; 1\le j<\alpha_m, \\
    \quo(a_mx_m^\ee{j-1}, p), &\; j=\alpha_m \text{ and } \alpha_m>0,
  \end{cases}
\end{align*}
is equivalent to
\begin{equation} \label{eq_Xk}
\begin{aligned}
  X_k &= \iota^k(X_{-1})Y_k +
  \rema\Big(
  \big[ \iota^{k-1}\rho(X_0) \,,\,
  \dots ,\, \iota^0\rho(X_{k-1})
  \big] U_k
  \,,\,p\Big) \\
  & \hspace*{2in}
  + \quo\Big(\big[ \iota^{k}(X_0) \,,\, \dots \,,\,
  \iota^1(X_{k-1})\big]U_k\,,\, p\Big),
\end{aligned}
\end{equation}
where $\iota,\rho:(M_s)^{l}\rightarrow(M_s)^{l+1}$ act column
by column, padding them with zeros:
\begin{equation}
  \iota(x)=[0;x], \qquad \rho(x)=[x;0], \qquad x\in (M_s)^{l}, \quad
  0\in M_s.
\end{equation}
Here $\Lambda\iota\Lambda^{-1}$ is multiplication by $p$, which embeds
$M_{ls}\cong M/p^lM$ in $M_{(l+1)s}\cong M/p^{l+1}M$ as a module over
$R$, while $\rho$ is an
embedding of vector spaces over $K$ (but not an $R$-module morphism).
If we define the matrices $\XX_0=X_0$ and
\begin{equation}
  \XX_k = [\iota(\XX_{k-1}),X_k] =
  \left[\begin{pmatrix} 0_{nk\times r_0} \\ X_0\end{pmatrix},
    \begin{pmatrix} 0_{n(k-1)\times r_1} \\ X_1 \end{pmatrix},
      \dots,\bigg( X_{k} \bigg)\right], \quad (k\ge1),
\end{equation}
then (\ref{eq_Xk}) simply becomes
\begin{equation}\label{eqn:Xk1}
  X_k = [\rema(\XX_{k-1}U_k, p);Y_k] + \quo(\iota(\XX_{k-1}U_k), p).
\end{equation}
As in (\ref{eq_yi1}) above, the matrix multiplications are done in the
ring $R$ before the quotient and remainder are computed to obtain
$X_k$.
Finally, we line up the columns of $X_{k-1}$ with the last
$r_{k-1}$ columns of $\dot{\mc{A}}_k$ and extract (i.e.~accept) columns of
$X_{k-1}$ that correspond to new, linearly independent columns of
$\dot{\mc{A}}_k$.  We denote the matrix of extracted columns by
$\wtil{X}_{k-1}$.  At the completion of the algorithm, the unimodular
matrix $V(\lambda)$ that puts $A(\lambda)$ in local Smith form
is given by
\begin{equation}
  V(\lambda) = \big[\Lambda(\wtil{X}_{-1}),\dots,\Lambda(\wtil{X}_{\beta-1})
  \big].
\end{equation}

The final algorithm is presented in Figure~\ref{fig:alg2}.  In the
step marked $\bullet$, we can avoid re-computing the reduced
row-echelon form of the first $n+R_{k-2}$ columns of $\dot{\mc{A}}_k$ by
storing the sequence of Gauss-Jordan transformations \cite{golub} that
reduced $\dot{\mc{A}}_{k-1}$ to row-echelon form.  To compute $[Y_k;U_k]$,
we need only apply these transformations to the new columns of $\dot{\mc{A}}_k$
and then proceed with the row-reduction algorithm on these final
columns.  Also, if $A_0$ is large and sparse, rather than reducing
to row-echelon form, one could find kernels using an $LU$ factorization
designed to handle singular matrices.
This would allow the use of graph theory (clique analysis) to choose
pivots in the Gaussian elimination procedure to minimize fill-in.  We
also note that if $\Delta(\lambda)$ contains only one irreducible
factor, the local Smith form is a (global) Smith form of $A(\lambda)$.

\begin{figure}[t]
\begin{center}
\includegraphics[scale=.65]{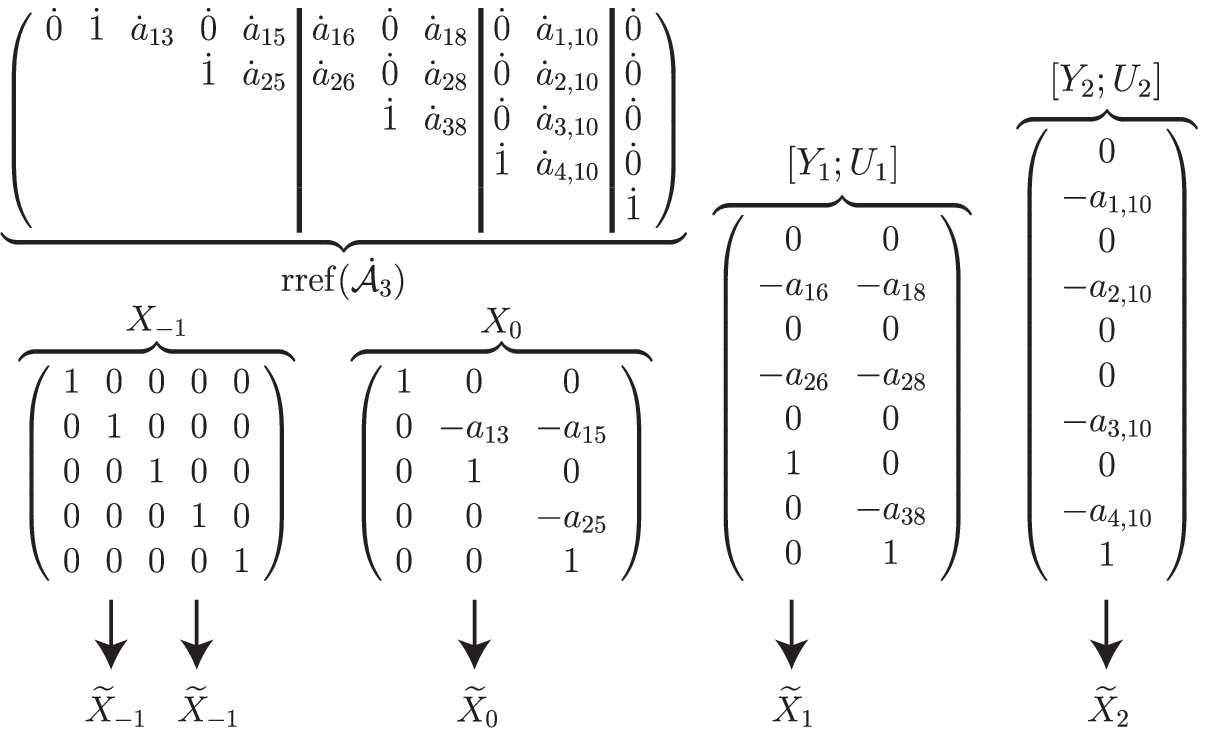}
\end{center}
\caption{The reduced row-echelon form of $\dot{\mc{A}}_\beta$ contains all
  the information necessary to construct $V(\lambda)
  =[\Lambda(\wtil{X}_{-1}),\dots,\Lambda(\wtil{X}_{s-1})]$.  An arrow
  from a column $[v;u]$ of $[Y_k;U_k]$ indicates that the vector
  $\big([\rema(\XX_{k-1}u, p);v]+\quo(\iota(\XX_{k-1}u), p)\big)$ is a column
  of $\wtil{X}_{k}$.}
\label{fig:extract}
\end{figure}

\begin{figure}[t]
\begin{center}
\fbox{\parbox{.9\linewidth}{
\begin{algori} \label{alg2} (Local Smith form, final version)
\vspace*{4pt}
\begin{tabbing}
\hspace*{.125in} \=
\hspace*{.02in} \= \hspace*{.125in} \= \hspace*{1.95in} \= \hspace*{.95in}
\= \kill
\> $k=0$, \; $R_{-1}=0$, \; $\mc{A}_0 = A^\ee0$ \\
\> $X_0 = \XX_0 = \opn{null}(\dot{\mc{A}}_0)$ \\
\> $r_0 = R_0 = \opn{num\_cols}(\XX_0)$ \>\>\> (number of columns) \\
\> $\wtil{X}_{-1}=[e_{j_1},\dots,e_{j_{n-r_0}}]$, \>\>\>
(columns $j_i$ of $\opn{rref}(\dot{\mc{A}}_0)$ start new rows) \\
\> \textbf{while} $R_k<\mu$ \>\>\> ($\mu$ = algebraic multiplicity of p) \\
\>\>\> $k=k+1$ \\
\>\>\> $\mc{A}_{k} = \big[ \mc{A}_{k-1} \, , \,
\rema(\big[A^\ee{k},\dots,A^\ee{1}\big]X_{k-1}, p) +
\quo(\big[A^\ee{k-1},\dots,A^\ee{0}\big]X_{k-1}, p) \, \big]$ \\
\>\>$\bullet$ \> $[Y_k;U_k]=\text{ new columns of }\opn{null}(\dot{\mc{A}}_{k})$
  beyond those of $\opn{null}(\dot{\mc{A}}_{k-1})$ \\
\>\>\> $r_k=\opn{num\_cols}(U_k)$, \>\> ($U_k$ is $R_{k-1}\times r_k$) \\
\>\>\> $R_k=R_{k-1}+r_k$ \\
\>\>\> $X_k=[\rema(\XX_{k-1}U_k ,p);Y_k] + \quo(\iota(\XX_{k-1}U_k), p)$
      \>\>($X_k$ is $n(k+1)\times r_k$) \\
\>\>\> $\XX_{k} = [\iota(\XX_{k-1}),X_k]$
      \>\>($\XX_k$ is $n(k+1)\times R_k$) \\
\>\>\> $\wtil{X}_{k-1}=X_{k-1}(:,[j_1,\dots,j_{r_{k-1}-r_k}])$, \>\>
(columns $n+R_{k-2}+j_i$ of  \;\;\\
\> \textbf{end while} \>\>\>\> \;
$\opn{rref}(\dot{\mc{A}}_k)$ start new rows) \\
\> $\beta=k+1$ \>\>\> (maximal Jordan chain length) \\
\> $\wtil{X}_{\beta-1}=X_{\beta-1}$ \\
\> $V(\lambda) = \big[\Lambda(\wtil{X}_{-1}),\dots,
  \Lambda(\wtil{X}_{\beta-1})\big]$
\end{tabbing}
\end{algori}
}}
\end{center}
\caption{Algorithm for computing a unimodular local Smith form.
}
\label{fig:alg2}
\end{figure}

\subsection{From Local to Global (Step 2)}\label{sec_step2}
Now that we have a local Smith form \eqref{eq_local_smith_j} for every
irreducible factor $p_j(\lambda)$ of $\Delta(\lambda)$, we can 
use the extended Euclidean algorithm to obtain a family of
polynomials $\{g_j(\lambda)\}_{j=1}^l$ with
$\deg(g_j(\lambda))<s_j\kappa_{jn}$, where $s_j=\deg(p_j)$,
such that
\begin{equation}\label{eq_gcd_n}
\sum_{j=1}^l\bigg[g_{j}(\lambda)\prod_{k=1,k\neq j}^l
p_k(\lambda)^{\kappa_{kn}}\bigg]=1,
\end{equation}
where $p_j(\lambda)^{\kappa_{jn}}$ is the last entry in the
diagonal matrix of the local Smith form at $p_j(\lambda)$. The
integers $\kappa_{jn}$ are positive. We define a matrix polynomial
$B_n(\lambda)$ via
\begin{equation}\label{eq_vi}
B_n(\lambda)=\sum_{j=1}^l\bigg[g_{j}(\lambda)
V_j(\lambda)\prod_{k\neq j}^l
p_k(\lambda)^{\kappa_{kn}}\bigg].
\end{equation}

The main result of this section is stated as follows.

\begin{proposition} \label{prop_tildeV}
The matrix polynomial $B_n(\lambda)$ in
\eqref{eq_vi} has two key properties:
\begin{enumerate}
\item Let ${b_{ni}}(\lambda)$ be the $i$th column of
  $B_n(\lambda)$. Then $A(\lambda){b_{ni}}(\lambda)$ is divisible by
  $d_i(\lambda)$, where
  $d_i(\lambda)=\prod_{j=1}^lp_j(\lambda)^{\kappa_{ji}}$ is the
  $i$th diagonal entry in $D(\lambda)$ of the Smith form.
\item $\det[B_n(\lambda)]$ is not divisible by $p_j(\lambda)$ for
  $j=1,\dots,l$.
\end{enumerate}
\end{proposition}

\begin{proof} (1) Let $v_{ji}(\lambda)$ be the $i$th column of
$V_j(\lambda)$. Then $A(\lambda)v_{ji}(\lambda)$
is divisible by $p_j(\lambda)^{\kappa_{ji}}$ and
\begin{equation*}
{b_{ni}}(\lambda)=\sum_{j=1}^l\bigg[\prod_{k\neq j}^l
p_k(\lambda)^{\kappa_{kn}}\bigg]g_{j}(\lambda)
v_{ji}(\lambda).
\end{equation*}
Since $\kappa_{jn}\geq\kappa_{ji}$ for $1\le i\le n$ and $1\le j\le l$,
$A(\lambda){b_{ni}}(\lambda)$ is divisible by $d_i(\lambda)$.

(2) The local Smith form construction ensures that
$p_j(\lambda) \nmid \det[V_j(\lambda)]$ for each $1\le j\le l$.
Equation (\ref{eq_gcd_n}) modulo $p_j(\lambda)$ shows that
$p_j(\lambda)\nmid g_j(\lambda)$.
By definition,
\begin{equation*}
  \begin{split}
    \det[B_n(\lambda)] &= \det\big(\big[
    b_{n1}(\lambda)\,,\,\dots\,,\,b_{nn}(\lambda)\big]\big) =
    \det\big(\big[b_{ni}(\lambda)\big]_{i=1}^n\big) \\
    &=\det\bigg(\bigg[
    \sum_{j'=1}^l \bigg(\prod_{k\ne j'}^lp_k(\lambda)^{\kappa_{kn}}
    \bigg)g_{j'}(\lambda)v_{j'i}(\lambda)\bigg]_{i=1}^n\bigg).
  \end{split}
\end{equation*}
Each term in the sum is divisible by $p_j(\lambda)$ except $j'=j$.
Thus, by multi-linearity,
\begin{equation*}
  \rema(\det[B_n(\lambda)], p_j(\lambda)) = \rema\Biggl(
  \bigg[\prod_{k\ne j}^l p_k(\lambda)^{\kappa_{kn}}\bigg]^n
  \big[g_j(\lambda)\big]^n\det\big[V_j(\lambda)\big],
  \,p_j(\lambda)\Biggr)\ne 0,
\end{equation*}
as claimed.
\end{proof}

\begin{rmk}
  It is possible for $\det[B_n(\lambda)]$ to be non-constant; however,
  its irreducible factors will be distinct from
  $p_1(\lambda),\dots,p_l(\lambda)$.
\end{rmk}

\begin{rmk}\label{rmk_mod_ltog} Rather than building $B_n(\lambda)$ as a linear
combination \eqref{eq_vi}, we may form $B_n(\lambda)$ with columns
\[b_{ni}(\lambda)=\sum_{j=1}^l\bigg[\prod_{k\neq j}^l
p_k(\lambda)^{\max(\kappa_{ki},1)}\bigg]g_{ij}(\lambda)
v_{ji}(\lambda),\qquad (1\le i\le n),
\]
where $\{g_{ij}\}_{j=1}^l$ solves the extended GCD problem
\[\sum_{j=1}^l\bigg[g_{ij}(\lambda)\prod_{k\neq j}^l
p_k(\lambda)^{\max(\kappa_{ki},1)}\bigg]=1.
\]
The two properties proved above also hold for this definition of
$B_n(\lambda)$. This modification can significantly reduce the size of
the coefficients in the computation when there is a wide range of
Jordan chain lengths.  But if $\kappa_{ji}$ only changes slightly for
$1\le i\le n$, this change will not significantly affect the total
running time of the algorithm.
\end{rmk}

\subsection{Construction of Unimodular Multipliers (Step
3)}\label{sec_step3}
Given $[f_1(\lambda);\dots;f_n(\lambda)]\in
K[\lambda]^n$, we can compute the Hermite form (\ref{eqn:hermite})
to obtain a unimodular matrix $Q(\lambda)$ such that, after
reversing rows,
$Q(\lambda)f(\lambda)=[0;\dots;0;r(\lambda)]$, where
$r=\gcd(f_1,\dots,f_n)$.  We apply this procedure to the last
column of $B_n(\lambda)$ and define $V_n(\lambda)=Q(\lambda)^{-1}$.
The resulting matrix
\[
  B_{n-1}(\lambda):=V_n(\lambda)^{-1}B_n(\lambda)
\]
is zero above the main diagonal in column $n$.  We then apply this
procedure to the first $n-1$ components of column $n-1$ of
$B_{n-1}(\lambda)$ to get a new $Q(\lambda)$, and define
\begin{equation}
V_{n-1}(\lambda)=\begin{pmat}({..|})
   &  &  & 0 \cr
   & Q(\lambda)^{-1} &  & \vdots \cr
   &  &  & 0 \cr\-
  0 & \cdots & 0 & 1 \cr
\end{pmat}.
\end{equation}
It follows that
$B_{n-2}(\lambda):=V_{n-1}(\lambda)^{-1}B_{n-1}(\lambda)$
is zero above the main diagonal in columns $n-1$ and $n$.
Continuing in this fashion, we obtain unimodular matrices
$V_n(\lambda)$, \dots, $V_2(\lambda)$ such that
\begin{equation*}
  A(\lambda)B_n(\lambda)=
  A(\lambda)\underbrace{V_n(\lambda)\cdots V_2(\lambda)}_{V(\lambda)}
  V_2(\lambda)^{-1}\cdots\underbrace{
    V_n(\lambda)^{-1}B_n(\lambda)}_{B_{n-1}(\lambda)} =
  A(\lambda)V(\lambda)B_1(\lambda),
\end{equation*}
where $V(\lambda)$ is unimodular, $B_1(\lambda)$ is lower triangular,
and
\begin{equation}\label{eq_detB1}
  \det[B_1(\lambda)]=\mathrm{const}\cdot \det[B_n(\lambda)].
\end{equation}
The matrix $V(\lambda)$ puts $A(\lambda)$ in Smith form:

\begin{proposition}\label{prop_uni}
There is a unimodular matrix polynomial $E(\lambda)$ such that
\begin{equation}\label{eq_AVED}
A(\lambda)V(\lambda)=E(\lambda)D(\lambda),
\end{equation}
where $D(\lambda)$ is of the form (\ref{eq_smith}).
\end{proposition}

\begin{proof}
  Let $r_{mi}(\lambda)$ denote the entry of $B_1(\lambda)$ in the
  $m$th row and $i$th column.  Define $y_i(\lambda)$ and
  $z_i(\lambda)$ to be the $i$th columns of $A(\lambda)V(\lambda)$ and
  $A(\lambda)V(\lambda)B_1(\lambda)$, respectively, so that
\begin{equation}\label{eq_zy}
  z_i(\lambda) = y_i(\lambda)r_{ii}(\lambda) +
  \sum_{m=i+1}^n y_m(\lambda)r_{mi}(\lambda), \qquad
  (1\le i\le n).
\end{equation}
By Proposition~\ref{prop_tildeV}, $z_i(\lambda)$ is divisible by
$d_i(\lambda)$ for $1\le i\le n$ and $p_j(\lambda)\nmid
\det[B_1(\lambda)]$ for $1\le j\le l$.  It follows that the diagonal
entries $r_{ii}(\lambda)$ of $B_1(\lambda)$ are relatively prime to
each of the $d_i(\lambda)$.  As $d_n(\lambda)$ divides
$y_n(\lambda)r_{nn}(\lambda)$ and is relatively prime to
$r_{nn}(\lambda)$, it divides $y_n(\lambda)$ alone.  Now suppose $1\le
i< n$ and we have shown that $d_m(\lambda)$ divides $y_m(\lambda)$ for
$i< m\le n$.  Then since $d_i(\lambda)$ divides $d_m(\lambda)$ for
$m>i$ and $r_{ii}(\lambda)$ is relatively prime to $d_i(\lambda)$, we
conclude from (\ref{eq_zy}) that $d_i(\lambda)$ divides
$y_i(\lambda)$.  By induction, $d_i(\lambda)$ divides $y_i(\lambda)$
for $1\le i\le n$.  Thus, there is a matrix polynomial $E(\lambda)$
such that (\ref{eq_AVED}) holds.  Because $V(\lambda)$ is unimodular
and $\det[A(\lambda)]=\opn{const}\cdot \det[D(\lambda)]$, it follows
that $E(\lambda)$ is also unimodular, as claimed.
\end{proof}

\begin{rmk} $V(\lambda)$ constructed as described above puts $A(\lambda)$
in a global Smith form whether we build $B_n(\lambda)$ as a linear
combination \eqref{eq_vi} or as in Remark~\ref{rmk_mod_ltog}.
\end{rmk}

\begin{rmk}\label{rmk_mod_uni_loop} We can stop before reaching
$V_2(\lambda)$ by adding a test
\begin{equation*}
  \jt\textbf{while }
  d_k\ne 1
\end{equation*}
to the loop in which $V(\lambda)$ is constructed.  When the loop
terminates, we have 
$V(\lambda)=V_n(\lambda)\cdots V_{k+1}(\lambda)$, where
$k$ is the largest integer for which
\begin{equation*}
d_1(\lambda) = \cdots = d_k(\lambda) = 1.
\end{equation*}
Note that $k$ is known from the local Smith form calculations.
The last $n-k$ columns of $V_n(\lambda)\cdots V_{k+1}(\lambda)$ are
the same as those of $V_n(\lambda)\cdots V_2(\lambda)$; therefore,
either can be used for $V(\lambda)$ as they contain identical Jordan
chains.
\end{rmk}

\begin{rmk}\label{rmk_Vk_Bk}
  A slight modification of this procedure can significantly reduce the
  degree of the polynomials and the size of the coefficients in the
  computation.  In this variant, rather than applying the extended GCD
  algorithm on $b_{nn}(\lambda)$ to find a unimodular matrix
  polynomial $Q(\lambda)$ so that $Q(\lambda) b_{nn}(\lambda)$ has the
  form $[0;\dots;0;r(\lambda)]$, we compute $Q(\lambda)$ that puts
  $\rema(b_{nn}(\lambda),d_n(\lambda))$ into this form. That is, we
  replace the last column of $B_n(\lambda)$ with
  $\rema(b_{nn}(\lambda),d_n(\lambda))$ before computing $Q(\lambda)$.
  To distinguish, we
  denote this new definition of $V_n(\lambda)=Q(\lambda)^{-1} $ by
  $\wtil{V}_n(\lambda)$ and the resulting $B_{n-1}(\lambda)$ by
  $\wtil{B}_{n-1}(\lambda)$. Continuing in this manner, we find
  unimodular matrix polynomials $\wtil{V}_n(\lambda),\dots,
  \wtil{V}_{k+1}(\lambda)$ by applying the procedure on
  $\rema(\tilde{b}_{ii}(\lambda),d_i(\lambda))$ for $i=n,\dots,k+1$,
  where $\tilde{b}_{ii}(\lambda)$ contains the first $i$ components of
  column $i$ of $\wtil{B}_i(\lambda)$ and $k$ is defined as in
  Remark~\ref{rmk_mod_uni_loop}. We also define
\begin{equation*}
  \bar{B}_i=\wtil{V}_{i+1}(\lambda)^{-1}\cdots\wtil{V}_n(\lambda)^{-1}B_n(\lambda),
  \qquad (k\leq i\leq n-1).
\end{equation*}
  Note that in general,
  $\bar{B}_i(\lambda)\neq\wtil{B}_i(\lambda)$. It remains
  to show that this definition of
  $\wtil{V}(\lambda)=\wtil{V}_n(\lambda)\dots\wtil{V}_{k+1}(\lambda)$,
  which satisfies
\begin{equation*}
  A(\lambda)B_n(\lambda)=
  A(\lambda)\underbrace{\wtil{V}_n(\lambda)\cdots
  \wtil{V}_{k+1}(\lambda)}_{\wtil{V}(\lambda)}
  \wtil{V}_{k+1}(\lambda)^{-1}\cdots\underbrace{
    \wtil{V}_n(\lambda)^{-1}B_n(\lambda)}_{\bar{B}_{n-1}(\lambda)} =
  A(\lambda)\wtil{V}(\lambda)\bar{B}_k(\lambda),
\end{equation*}
also puts $A(\lambda)$ in Smith form:
\begin{proposition}
There is a unimodular matrix polynomial $\wtil{E}(\lambda)$ such that
\begin{equation}\label{eq_mod_AVED}
A(\lambda)\wtil{V}(\lambda)=\wtil{E}(\lambda)D(\lambda),
\end{equation}
where $D(\lambda)$ is of the form (\ref{eq_smith}).
\end{proposition}
\begin{proof} Define $\tilde{q}_i(\lambda)=\left[\quo(\tilde{b}_{ii}(\lambda),d_i(\lambda));0\right]\in M=R^n$ for
$i=n,\dots,k+1$, where $0\in R^{n-i}$, $\tilde{b}_{ii}(\lambda)$ 
was defined above,
and $\wtil{B}_n(\lambda):=B_n(\lambda)$. Then we have
\begin{equation*}\begin{split}
\wtil{B}_{n-1}(\lambda)&=\wtil{V}_n(\lambda)^{-1}\left(B_n(\lambda)-\left[\begin{array}{c|c}0_{n\times(n-1)}&d_n(\lambda)\tilde{q}_n(\lambda)\end{array}\right]\right)\\
&=\bar{B}_{n-1}(\lambda)-\left[\begin{array}{c|c}0_{n\times(n-1)}&d_n(\lambda)\wtil{V}_n(\lambda)^{-1}\wtil{q}_n(\lambda)\end{array}\right].
\end{split}
\end{equation*}
The first $n-1$ columns of $\wtil{B}_n(\lambda)$ are the same as those of $\bar{B}_n(\lambda)$.
Continuing,
we have
\begin{equation*}
\begin{split}
&\wtil{B}_{n-2}(\lambda)=\wtil{V}_{n-1}(\lambda)^{-1}\left(\wtil{B}_{n-1}(\lambda)-\left[\begin{array}{c|c|c}0_{n\times(n-2)}&d_{n-1}(\lambda) \tilde{q}_{n-1}(\lambda)&0_{n\times 1}\end{array}\right]\right)\\
&\;\;=\wtil{V}_{n-1}(\lambda)^{-1}\left(\bar{B}_{n-1}(\lambda)-\left[\begin{array}{c|c|c}0_{n\times(n-2)}&d_{n-1}(\lambda) \tilde{q}_{n-1}(\lambda)&d_n(\lambda)\wtil{V}_n(\lambda)^{-1}\tilde{q}_n(\lambda)\end{array}\right]\right)\\
&\;\;=\bar{B}_{n-2}(\lambda)-\left[\begin{array}{c|c|c}0_{n\times(n-2)}&d_{n-1}(\lambda)\wtil{V}_{n-1}(\lambda)^{-1}\tilde{q}_{n-1}(\lambda)&d_n(\lambda)\wtil{V}_{n-1}(\lambda)^{-1}\wtil{V}_n(\lambda)^{-1}\tilde{q}_n(\lambda)\end{array}\right].
\end{split}
\end{equation*}
It follows by induction that
\begin{align}\label{eqn:Btilde:vs:Bbar}
&\wtil{B}_{k}(\lambda)=\wtil{V}_{k+1}(\lambda)^{-1}\left(\wtil{B}_{k+1}(\lambda)-\left[\begin{array}{c|c|c}0_{n\times k}&d_{k+1}(\lambda) \tilde{q}_{k+1}(\lambda)&0_{n\times (n-k-1)}\end{array}\right]\right)\\
\notag
&\;\;=\bar{B}_{k}(\lambda)-\left[\begin{array}{c|c|c|c}0_{n\times k}&d_{k+1}(\lambda)\wtil{V}_{k+1}(\lambda)^{-1}\tilde{q}_{k+1}(\lambda)&\cdots&d_n(\lambda)\wtil{V}_{k+1}(\lambda)^{-1}\cdots\wtil{V}_n(\lambda)^{-1}\tilde{q}_n(\lambda)\end{array}\right].
\end{align}
$\wtil{B}_k(\lambda)$ is zero above the main diagonal in columns $k+1$ to $n$. Define 
\begin{equation*}
u_i(\lambda):=\wtil{V}_{k+1}(\lambda)^{-1}\cdots \wtil{V}_i(\lambda)^{-1}\tilde{q}_i(\lambda),
\qquad (k+1\le i\le n).
\end{equation*}
Then the $i$th column of the difference $\bar{B}_k(\lambda)-\wtil{B}_k(\lambda)$ is $d_i(\lambda)u_i(\lambda)$
for $k+1\leq i\leq n$.

Let $\tilde{r}_{mi}(\lambda)$ denote the entry of $\wtil{B}_k(\lambda)$ in the $m$th row and $i$th column. Define $\tilde{y}_i(\lambda)$ and $z_i(\lambda)$ to be the
  $i$th columns of $A(\lambda)\wtil{V}(\lambda)$ and
  $A(\lambda)\wtil{V}(\lambda)\bar{B}_k(\lambda)$, respectively, so that
\begin{equation*}
 z_i(\lambda) = \bigg[ \tilde{y}_i(\lambda)\tilde{r}_{ii}(\lambda) +
  \sum_{m=i+1}^n \tilde{y}_m(\lambda)\tilde{r}_{mi}(\lambda)\bigg] +
  d_i(\lambda)A(\lambda)\wtil{V}(\lambda)u_i(\lambda), \qquad
  (k+1\le i\le n).
\end{equation*}
By Proposition~\ref{prop_tildeV}, $z_i(\lambda)$ is divisible by
$d_i(\lambda)$ and $p_j(\lambda)\nmid
\det[B_n(\lambda)]=\mathrm{const}\cdot\det[\bar{B}_{i-1}(\lambda)]$ for
$k+1\le i\le n$ and $1\le j\le l$.  As $d_i$ divides $d_m$ for $i\leq
m\leq n$ and since (\ref{eqn:Btilde:vs:Bbar}) holds with $k$ replaced
by $i-1$ for $k+1\le i\le n$,
$\det[\bar{B}_{i-1}(\lambda)]-\det[\wtil{B}_{i-1}(\lambda)]$ is
divisible by $d_{i}(\lambda)$ due to multi-linearity of
determinants. We also know that $\det[\wtil{B}_{i-1}(\lambda)]$ is
divisible by $\tilde{r}_{ii}(\lambda)$. Proof by contradiction shows
that $\tilde{r}_{ii}(\lambda)$ is relatively prime to $d_i(\lambda)$
for $k+1\le i\le n$.  Then we argue by induction as in the proof of
Proposition~\ref{prop_uni} to conclude that $d_i(\lambda)$ divides
$\tilde{y}_i(\lambda)$ for $k+1\le i\le n$. It holds trivially for $1\leq
i\leq k$ as $d_1=\dots=d_k=1$. Thus, there is a matrix polynomial
$\wtil{E}(\lambda)$ such that (\ref{eq_mod_AVED}) holds.  Because
$\wtil{V}(\lambda)$ is unimodular and
$\det[A(\lambda)]=\opn{const}\cdot \det[D(\lambda)]$, it follows that
$\wtil{E}(\lambda)$ is also unimodular.
\end{proof}
\end{rmk}

\section{Performance Comparison}\label{sec_numerical}

In this section, we compare our algorithm to Villard's method with
good conditioning \cite{Villard95}, which is another deterministic
sequential method for computing Smith forms with multipliers, and to
`{\sf MatrixPolynomialAlgebra[SmithForm]}' in Maple.  All three
algorithms are implemented in exact arithmetic using Maple 13.  The
maximum number of digits that Maple can use for the numerator and
denominator of a rational number (given by `{\sf
  kernelopts(maxdigits)}') is over 38 billion.
However, limitations of available memory and
running time set the limit on the largest integer number much lower
than this.  We use the variant of Algorithm~\ref{alg1} given in
Appendix~\ref{appendix} to compute local Smith forms.

To evaluate the performance of these methods, we generate several
groups of diagonal matrices $D(\lambda)$ over $\QQ$ and multiply them
on each side by unimodular matrices of the form
$L(\lambda)Z(\lambda)$, where $L(\lambda)$ is unit lower triangular and
$Z(\lambda)$ is
unit upper triangular, both with off diagonal entries of the form
$\lambda-i$ with $i\in\{-10,\dots,10\}$ a random integer. As a final
step, we apply a row or column permutation to the resulting matrix.
We find that row permutation has little effect on the running time of
the algorithms while column permutation reduces the performance of
Villard's method.
We compare the results in two extreme cases: (1) without column
permutation and (2) with columns reversed. Each process is repeated
five times for each $D(\lambda)$ and the median running time is
recorded.

We use several parameters in the comparison, including the size $n$ of
the square matrix $A(\lambda)$, the bound $d$ of the polynomial
degrees of the entries in $A(\lambda)$, the number $l$ of irreducible
factors in $\det[A(\lambda)]$, and the maximal Jordan chain length
$\kappa_{jn}$.

In Figure~\ref{fig:test1_3}, we show the running time of three tests
with linear irreducible factors of the form $p_j=\lambda-\lambda_j$.
As Villard's method and Maple compute the left and right multipliers
$U(\lambda)$ and $V(\lambda)$ while our algorithm instead computes
$E(\lambda)$ and $V(\lambda)$, we also report the cost of
inverting $E(\lambda)$ to obtain $U(\lambda)$ at the end of our
algorithm (using Maple's matrix inverse routine).  This step could be
made significantly faster by taking advantage of the fact that
$E(\lambda)$ is unimodular.
For example, one could store the
sequence of elementary unimodular operations such that
  $T(\lambda)=Q_m(\lambda)\cdots Q_1(\lambda) E(\lambda)$
is unit upper triangular.
It would not be necessary to actually form
the matrices $T(\lambda)^{-1}$ or
\begin{equation}\label{eqn:UTQ}
  U(\lambda)=T(\lambda)^{-1}Q_m(\lambda)\cdots Q_1(\lambda)
\end{equation}
as the right hand side can be applied directly to any vector
polynomial using back substitution to solve $T(\lambda)x(\lambda)=
z(\lambda)$ in the last step.
The
same idea is standard in numerical linear algebra, where the
$LU$-decomposition of a matrix is less expensive to compute than its
inverse, and is equally useful.
In the first test of Figure~\ref{fig:test1_3}, $D_n(\lambda)$ is of the form
\[D_n(\lambda)=\opn{diag}[1,\dots,1,\lambda,\lambda(\lambda-1),
\lambda^2(\lambda-1),\lambda^2(\lambda-1)^2],
\]
where the matrix size $n$ increases, starting with $n=4$. Hence, we
have $d=8$, $l=2$, and $\kappa_{1n}=\kappa_{2n}=2$ all fixed. (The
unimodular matrices in the construction of $A(\lambda)$ each have
degree 2.)  For this test, inverting $E(\lambda)$ to obtain
$U(\lambda)$ is the most expensive step of our algorithm.  Without
column permutation of the test matrices, our algorithm (with
$U(\lambda)$) and Villard's method have similar running times, both
outperforming Maple's built-in function.  With column permutation, the
performance of Villard's method drops to the level of Maple's routine
while our algorithm remains faster.
\begin{figure}[t]
\begin{center}
\includegraphics[width=\linewidth,trim=0 0 0 0,clip]{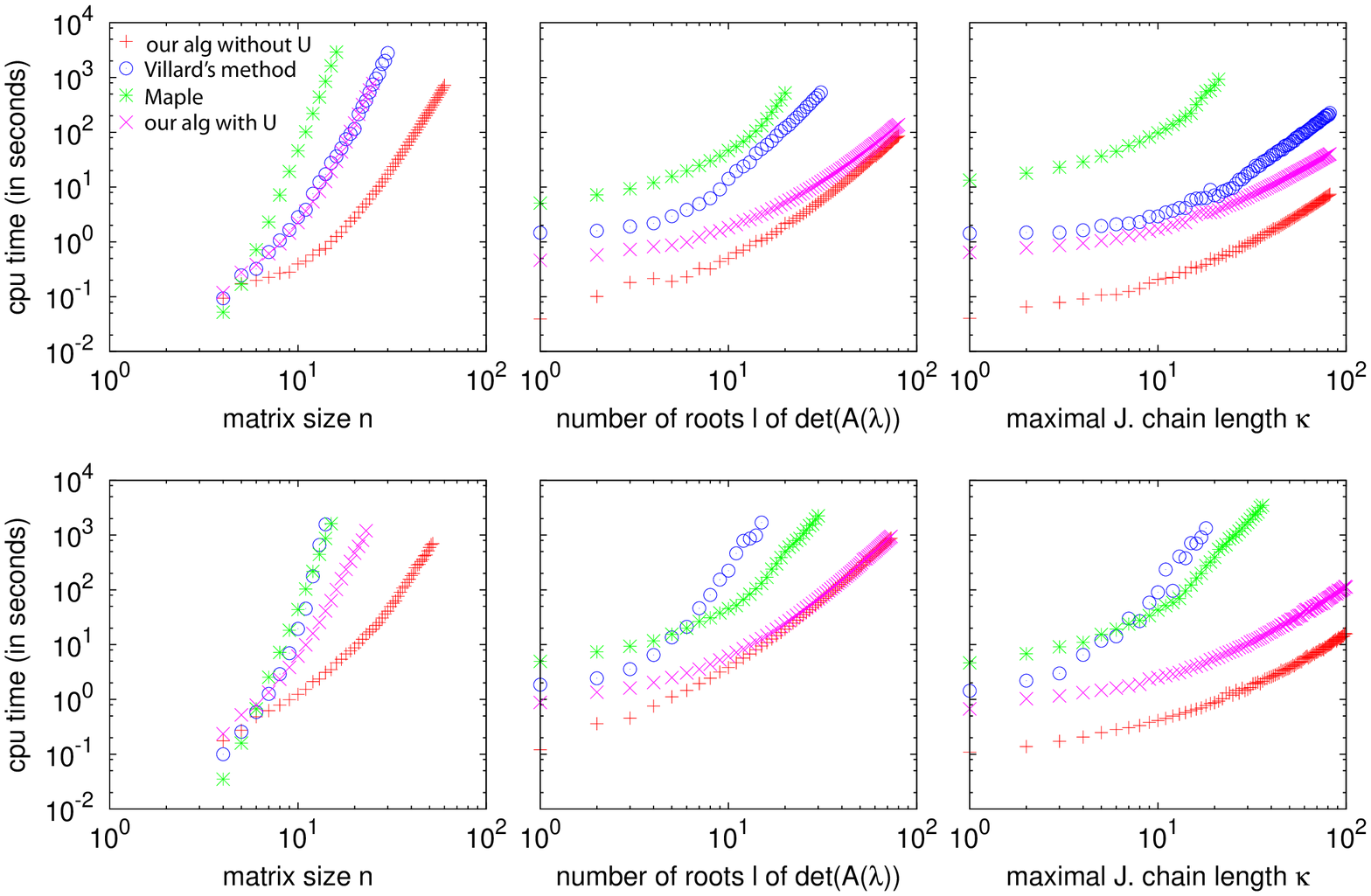}
\caption{ Comparison of running time of our algorithm (with or without
  computing $U(\lambda)$) to Villard's method, and to Maple's Smith
  form routine, on three families of test matrices.  (Top row) without
  column permutation of test matrices.  (Bottom row) with column
  permutation.  }
\label{fig:test1_3}
\end{center}
\end{figure}
For the second test, we use test matrices $D_l(\lambda)$ of size
$9\times 9$, where $l$ is the number of roots of $\det[A(\lambda)]$:
\[D_l(\lambda)=\opn{diag}[1,\dots,1,\prod_{j=1}^l(\lambda-j)],
 \qquad (l=1,2,\dots).
\]
Thus, $n=9$,
$d=l+4$ and $\kappa_{jn}=1$ for $1\le j\le l$.  This time the
relative cost of inverting $E(\lambda)$ to obtain $U(\lambda)$
decreases with $l$ in our algorithm, which is significantly faster
than the other two methods whether or not we permute columns
in the test matrices.
In the third test, we use $9\times 9$ test matrices $D_k(\lambda)$ of
the form
\[D_k(\lambda)=\opn{diag}[1,\dots,1,(\lambda-1)^k], \qquad (k=1,2,\dots),
\]
with $n=9$, $l=1$, $\kappa_{1n}=k$ and $d=k+4$.
We did not implement the re-use strategy for
computing the reduced row-echelon form of $\mc{A}_k$ by storing the
Gauss-Jordan transformations used to obtain
$\opn{rref}(\mc{A}_{k-1})$, and then continuing with only the new
columns of~$\mc{A}_k$.  This is because the built-in function {\sf
  LinearAlgebra[ReducedRowEchelonForm]} is much faster than can be
achieved by a user defined Maple code for the same purpose.  In a
lower level language (or with access to Maple's internal code), this
re-use strategy would decrease the running time of local
Smith form calculations in this test from $O(k^4)$ to $O(k^3)$.
A similar decrease in the cost of computing the left-multiplier
$U(\lambda)=E(\lambda)^{-1}$ could be achieved by computing
$T(\lambda)$ in (\ref{eqn:UTQ}) instead.

\begin{figure}[t]
\begin{center}
\includegraphics[width=\linewidth,trim=0 0 0 0,clip]{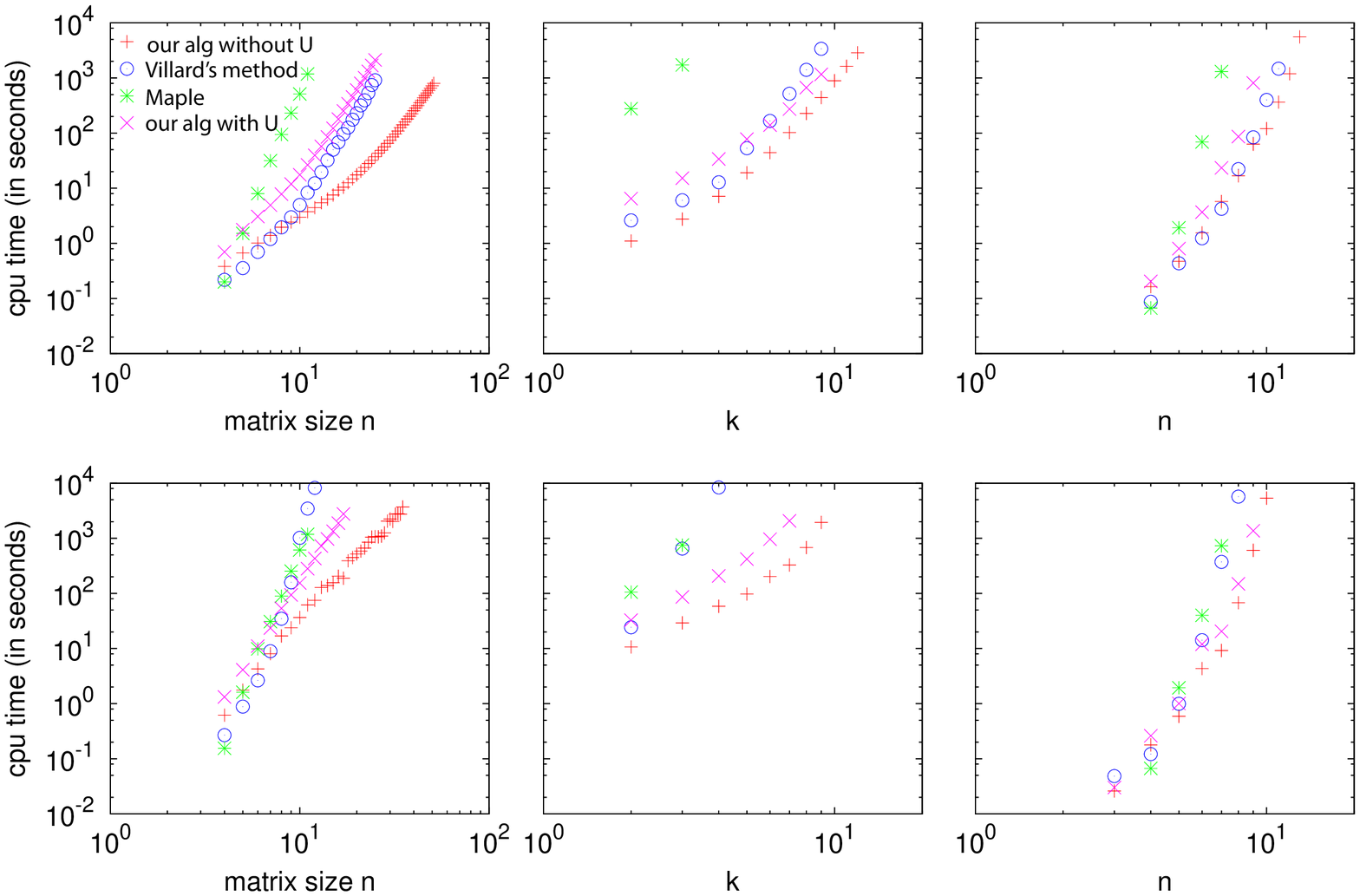}
\caption{ Comparison of running times of the algorithms for three
  test problems in which the irreducible factors $p_j(\lambda)$ of the
  determinant are of degree greater than 1.  (Top row) without column
  permutation of test matrices.  (Bottom row) with column permutation.
}
\label{fig:test4_6}
\end{center}
\end{figure}

We also evaluate the performance on three test problems (numbered
4--6) with irreducible polynomials of higher degree.  The results are
given in Figure~\ref{fig:test4_6}. In the fourth test, we use matrices
$D_n(\lambda)$ similar to those in the first test, but with
irreducible polynomials of degree $2$ and $4$.  Specifically, we
define
\[D_n(\lambda)=\opn{diag}[1,\dots,1,p_1, p_1 p_2,p_1^2 p_2,p_1^2p_2^2],
\qquad (n=4,5,\dots),
\]
where $p_1=\lambda^2+\lambda+1$,
$p_2=\lambda^4+\lambda^3+\lambda^2+1$, $\kappa_{1n}=2$,
$\kappa_{2n}=2$, and $d=16$.  When the columns of the test matrices
are permuted, our algorithm is faster than the other two methods
whether or not $U(\lambda)$ is computed.  When the columns are not
permuted, computing $U(\lambda)$ causes our method to be slower than
Villard's method.  In this test, our algorithm would benefit from
switching to the $R/pR$ version of Algorithm~\ref{alg2} rather than
the version over $K$ described in the appendix.  It would also benefit
from computing $T(\lambda)$ in (\ref{eqn:UTQ}) rather than the full
inverse $U(\lambda)=E(\lambda)^{-1}$.
In the fifth test, we use $9\times 9$
test matrices $D_k(\lambda)$ of the form
\[D_k(\lambda)=\opn{diag}[1,\dots,1,\prod_{j=1}^k (\lambda^2+j),
\prod_{j=1}^k (\lambda^2+j)^2, \prod_{j=1}^k (\lambda^2+j)^k ],
\qquad (k=2,3,\dots),
\]
with $n=9$, $l=k$, $\kappa_{jn}=k$ and $d=2k^2+4$. Both the number of
factors and maximal Jordan chain length increase with $k$.  Our
algorithm performs much better than the others when column
permutations are performed on the test matrices.  In the final test,
we define $n\times n$ matrices
\[D_n(\lambda)=\opn{diag} [1,1, (\lambda^2+1), (\lambda^2+1)^2(\lambda^2+2),
\dots, \prod_{j=1}^{n-2}(\lambda^2+j)^{n-1-j}], \qquad
(n=3,4,\dots)
\]
so that all the parameters $n$, $l=n-2$, $\kappa_{jn}=n-1-j$ and
$d=(n-1)(n-2)+4$ increase simultaneously.  All three algorithms run
very slowly on this last family of test problems.

\section{Discussion}\label{sec_discussion}

The key idea of our algorithm is that it is much less expensive to
compute local Smith forms through a sequence of nullspace calculations
than it is to compute global Smith forms through a sequence of
unimodular row and column operations.  This is because (1) row
reduction over $R/pR$ in Algorithm~\ref{alg2} (or over $K$ in
Appendix~\ref{appendix}) is less expensive than computing B\'{e}zout
coefficients over $R$; (2) the size of the rational numbers that occur
in the algorithm remain smaller (as we only deal with the leading
terms of $A$ in an expansion in powers of $p$ rather than with all of
$A$); and (3) each column of $V(\lambda)$ in a local Smith form only
has to be processed once for each power of $p$ in the corresponding
diagonal entry of $D(\lambda)$.  Once the local Smith forms are known,
we combine them to form a (global) multiplier $V(\lambda)$ for
$A(\lambda)$.  This last step does involve triangularization of
$B_n(\lambda)$ via the extended GCD algorithm, but this is less time
consuming in most cases than performing unimodular row and column
operations on $A(\lambda)$ to obtain $D(\lambda)$.  This is because
we only have to apply row operations to $B_n(\lambda)$ (as the
columns are already correctly ordered); we keep the degree of
polynomials (and therefore the number of terms) in the algorithm small
with the operation $\rema(\cdot,d_i)$; and the leading columns of
$B_n(\lambda)$ tend to be sparse (as they consist of a superposition
of local Smith forms, whose initial columns $X_{-1}$ are a subset of
the columns of the identity matrix).  Sparsity is not used explicitly
in our code, but it does reduce the work required to compute the
B\'{e}zout coefficients of a column.

\begin{figure}[t]
\begin{center}
\includegraphics[width=\linewidth,trim=0 0 0 0,clip]{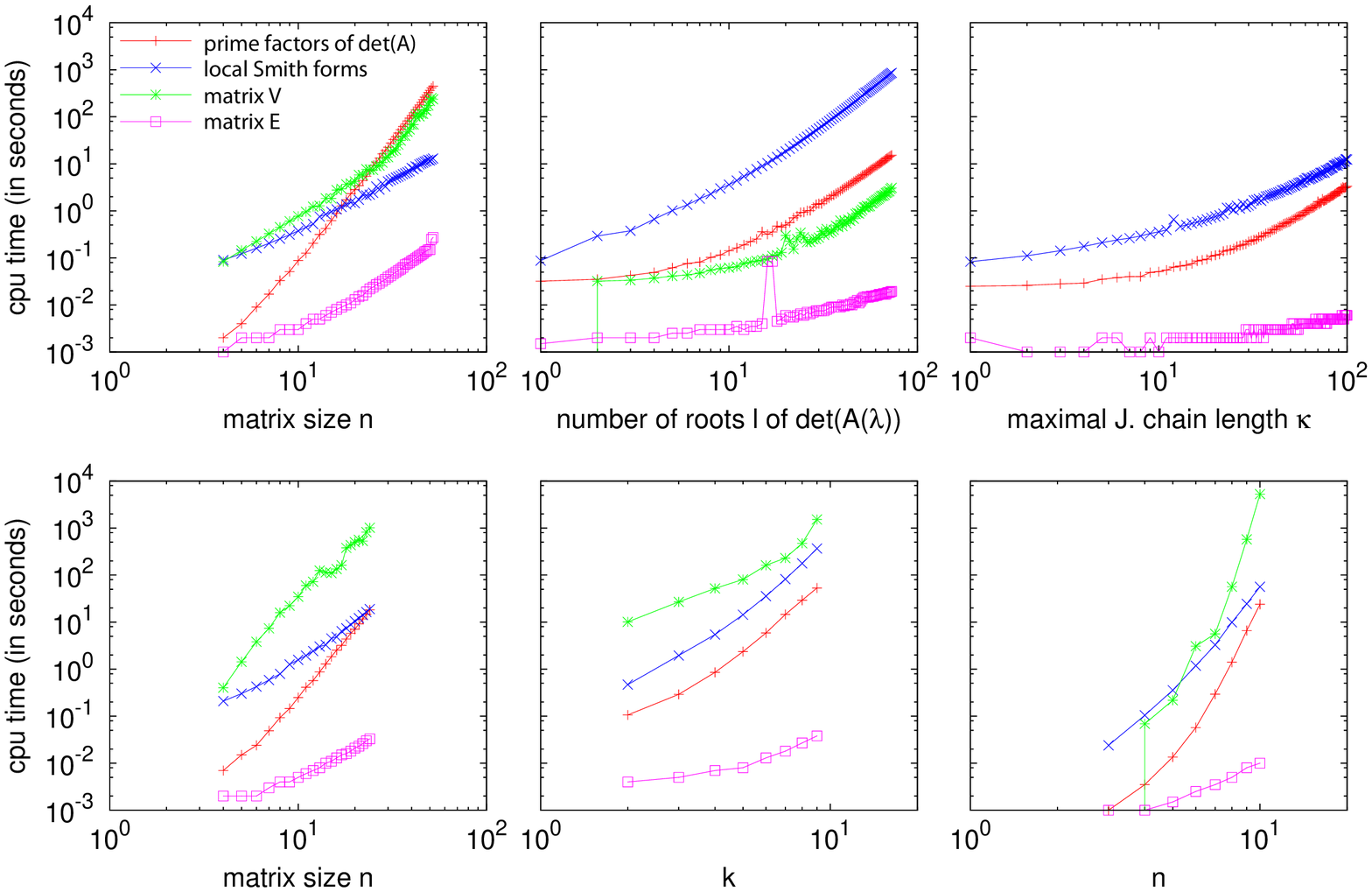}
\caption{Running time of each step of our algorithm
  for the six test problems of Section~\ref{sec_numerical}.  }
\label{fig:breakdown}
\end{center}
\end{figure}

A detailed breakdown of the running time of each step of our algorithm
is given in Figure~\ref{fig:breakdown}.  For each test in
Section~\ref{sec_numerical}, we show only the case where columns of
the test matrices are permuted; the other case is similar.  The step
labeled ``prime factors of $\opn{det}(A)$'' shows the time of
computing the determinant and factoring it into prime factors.  The
step labeled ``local Smith forms'' could be made faster in tests 4--6
by working over $R/pR$ (using Algorithm~\ref{alg2} rather than the
variant in the appendix) as the irreducible factors $p_j(\lambda)$
have degree $s_j>1$ in these tests.  Also, although it is not
implemented in this paper, this local Smith form construction would be
easy to parallelize.  The step labeled ``matrix $V$'' reports the time
of computing $V(\lambda)$ from $B_n(\lambda)$.  The cost of this step
is zero when there is only one irreducible factor in
$\det[A(\lambda)]$ as $B_n(\lambda)$ is already unimodular in that
case.  This happens when $l=1$ in the second test, in all cases in the
third test, and when $n=3$ in the last test.  Finally, the step
labeled ``matrix $E$'' reports the time of computing
$E(\lambda)=A(\lambda)V(\lambda)D(\lambda)^{-1}$.

The obvious drawback of our algorithm is that we have to compute a
local Smith form for each irreducible factor of $\Delta(\lambda)$
separately, while much of the work in deciding whether to accept a
column in Algorithm~\ref{alg1} can be done for all the irreducible
factors simultaneously by using extended GCDs.  In our numerical
experiments, it appears that in most cases, the benefit of computing
local Smith forms outweighs the fact that there are several of them to
compute.

\appendix

\section{Alternative Version of Algorithm~\ref{alg2}}
\label{appendix}

In this section we present an algebraic framework for local Smith
forms of matrix polynomials that shows the connection between
Algorithm~\ref{alg2} and the construction of canonical systems of
Jordan chains presented in \cite{Wilken2007}.  This leads to a variant
of the algorithm in which row-reduction is done in the field
$K$ rather than in $R/pR$.

Suppose $R$ is a principal ideal domain and $p$ is a prime in $R$.
$M$ defined via $M=R^n$ is a free $R$-module with a free basis
$\{(1,0,\dots,0),\dots,(0,\dots,1)\}$. Suppose $A:M\to M$ is a
$R$-module morphism. We define submodules
\begin{equation}
  N_k = \big\{ x\in M: Ax \text{ is divisible by }p^{k}\big\}, \qquad
    (k\ge 0).
\end{equation}
Then $N_k$ is a free submodule of $M$ by the structure theorem
\cite{hungerford} for finitely generated modules over a principal
ideal domain. (The structure theorem states that if $M$ is a free
module over a principal ideal domain $R$, then every submodule of $M$
is free.) The rank of $N_k$ is also $n$, as $p^{k}M\subset
N_k\subset M$.  Note that $N_0=M$ and
\begin{alignat}{2}
\label{eq_Nk1}
  N_{k+1} &\subset N_k, &\qquad &(k\ge0), \\
\label{eq_Nk2}
  N_{k+1} \cap pM &= pN_k, &\qquad &(k\ge0).
\end{alignat}
Next we define the spaces $W_k$ via
\begin{equation}
    W_k = N_{k+1}/pN_k, \qquad (k\ge-1),
\end{equation}
where $N_{-1}:=M$ so that $W_{-1}=M/pM$.  By (\ref{eq_Nk2}), the
action of $R/pR$ on $W_k$ is well-defined, i.e.~$W_k$ is a vector
space over this field.  Let us denote the canonical projection
$M\rightarrow M/pM$ by $\pi$.  Note that $\pi(pN_k)=0$, so $\pi$
is well-defined from $W_k$ to $M/pM$ for $k\ge-1$.  It is also
injective as $xp\in N_{k+1} \Rightarrow x\in N_k$, by (\ref{eq_Nk2}).
Thus, cosets $\{\dot{x}_1,\dots,\dot{x}_m\}$ are linearly independent
in $W_k$ iff $\{\pi(x_1),\dots,\pi(x_m)\}$ are linearly independent in
$M/pM$.  We define the integers
\begin{equation}
  r_k=\text{dimension of $W_k$ over $R/pR$}, \qquad (k\ge-1)
\end{equation}
and note that $r_{-1}=n$ and $r_k>0$ iff there exists $x\in
M$ such that $p\nmid x$ and $p^{k+1}\mid Ax$.  We also observe that
the truncation operator
\begin{equation}
  id:W_{k+1}\rightarrow W_k:(x+pN_{k+1})\mapsto(x+pN_{k}),
  \qquad (k\ge-1)
\end{equation}
is well-defined ($pN_{k+1}\subset pN_{k}$) and injective
($x\in N_{k+2}$ and $x\in pN_{k}$ $\Rightarrow$ $x\in pN_{k+1}$,
due to (\ref{eq_Nk2})).  We may therefore consider
$W_{k+1}$ to be a subspace of $W_k$ for $k\ge-1$, and
have the inequalities
\begin{equation}
  r_{k+1}\le r_k, \qquad (k\ge-1).
\end{equation}
The case $r_0=0$ is not interesting (as $N_k=p^{k}M$ for $k\ge0$), so
we assume that $r_0>0$.  Lemma~\ref{lem:det} shows that when
$R=K[\lambda]$, which we assume from now on, $r_0>0$ is equivalent to
the condition that $\det[A(\lambda)]$ is divisible by $p(\lambda)$.
We also assume that $r_k$ eventually decreases to zero, say
\begin{equation}
  r_k=0 \quad \Leftrightarrow \quad k\ge\beta, \qquad
  \beta:=\text{maximal Jordan chain length}.
\end{equation}
This follows from the assumption that $\det[A(\lambda)]$ is not
identically zero.  It will be useful to define the index sets
$I_k = \{i\;:\;n-r_{k-1}+1\le i\le n-r_k\}$ for $k=0,\dots,\beta$.

Any matrix $V=[x_1,\dots,x_n]$ will yield a local Smith form $AV=ED$
provided that
$x_i\in N_k$ for $i\in I_k$ ($0\le k\le \beta$) and the vectors
\begin{equation}\label{eq_basis}
  \{x_i+pN_{k-1}\}_{i\in I_k}
\end{equation}
form a basis for any complement $\wtil{W}_{k-1}$ of $W_k$ in $W_{k-1}$.
To see that $p\nmid\det E$, we use induction on $k$ to show that the vectors
\begin{equation}
  \{\quo(Ax_i, p^{\alpha_i})+pM\}_{i=1}^{n-r_k}
\end{equation}
are linearly independent in $M/pM$, where
\begin{equation}\label{eq_airk}
  \alpha_i = k, \qquad (i\in I_k).
\end{equation}
Otherwise, a linear combination of the form $\star$ in
Algorithm~\ref{alg1} would exist that belongs to $\wtil{W}_{k-1}\cap
W_k$, a contradiction.  The result that $p\nmid\det E$ follows
from Lemma~\ref{lem:det}.
The \emph{while} loop in Algorithm~\ref{alg1}
is a systematic procedure for computing such a collection
$\{x_i\}_{i\in I_k}$, and has the added benefit of yielding a
unimodular multiplier $V$.

We now wish to find a convenient representation for these spaces
suitable for computation.  Since $p^{k+1}M\subset pN_{k}$, we have
the $R$-module isomorphism
\begin{equation}
  N_{k+1}/pN_{k}\cong(N_{k+1}/p^{k+1}M)/(pN_{k}/p^{k+1}M),
\end{equation}
i.e.
\begin{equation}\label{eq_WfromWW}
  W_k\cong\WW_k/p\WW_{k-1}, \quad (k\ge0), \qquad
  \WW_k:=N_{k+1}/p^{k+1}M, \quad  (k\ge-1).
\end{equation}
Although the quotient $\WW_k/p\WW_{k-1}$ is a vector space over
$R/pR$, the spaces $\WW_k$ and $M/p^{k+1}M$ are not.
They are, however, modules over $R/p^{k+1}R$ and
vector spaces over $K$.
Note that
$A(\lambda)$ induces a linear operator
$\AAA_k$ on $M/p^{k+1}M$ with kernel
\begin{equation} \label{eq_WkerA}
  \WW_k = \ker\AAA_k, \qquad (k\ge-1).
\end{equation}
We also define
\begin{equation}
  R_k = \frac{\text{dimension of $\WW_k$ over } K}{s}, \qquad
  (k\ge-1,\quad s=\deg p)
\end{equation}
so that $R_{-1}=0$ and
\begin{equation}
  R_k=r_0+\cdots+r_k, \qquad (k\ge0),
\end{equation}
where we used $\WW_0=W_0$ together with (\ref{eq_WfromWW}) and the
fact that as a vector space over $K$, $\dim W_k=sr_k$.  By
(\ref{eq_airk}), $r_{k-1}-r_k=\#\{i\;:\;\alpha_i=k\}$, so
\begin{equation} \label{eq_alg_mult}
\begin{aligned}
  R_{\beta-1} &= r_0+\cdots+r_{\beta-1} =
  (r_{-1}-r_0)0 + (r_0-r_1)1 + \cdots + (r_{\beta-1}-r_\beta)\beta \\
  &= \alpha_1 + \cdots + \alpha_n = \mu =
  \text{algebraic multiplicity of $p$},
\end{aligned}
\end{equation}
where we used Theorem~\ref{thm:equiv} in the last step.  We also note
that $\nu:=R_0=s^{-1}\dim\ker(\AAA_0)$ can be interpreted as the
geometric multiplicity of $p$.

Equations (\ref{eq_WfromWW}) and (\ref{eq_WkerA}) reduce the problem
of computing Jordan chains to that of finding kernels of the linear
operators $\AAA_k$ over $K$.  If we represent elements $x\in M/p^{k+1}M$
as lists of coefficients $x^\ee{j,l,m}\in K$ such that
the components of $x$ involve the terms
\begin{equation}\label{eqn:enum:x}
  x^\ee{j,l,m}p^j\lambda^m, \qquad 0\le j\le k, \quad
  1 \le l\le n, \quad 0\le m\le s-1,
\end{equation}
then multiplication by $\lambda$ in $M/p^{k+1}M$ becomes the
following linear operator on $K^{sn(k+1)}$:
\begin{equation}\label{eqn:SSSk}
  \SSS_k = \begin{pmatrix}
    I\otimes S & 0 & 0 & 0 \\[1pt]
    I\otimes Z & I\otimes S & 0 & 0 \\[-4pt]
    0  & \ddots & \ddots & 0\\[1pt]
    0 & 0 & I\otimes Z & I\otimes S
  \end{pmatrix}\!, \qquad
  S \text{ as in (\ref{eq_S_def})}, \qquad
    Z = \begin{pmatrix} 0 & 0 & 1 \\[-4pt] & \ddots & 0 \\[2pt]
      0 & & 0
      \end{pmatrix}\!.
\end{equation}
Here $\SSS_k$ is a $(k+1)\times(k+1)$ block matrix, $I\otimes S$ is a
Kronecker product of matrices, $S$ and $Z$ are $s\times s$ matrices,
and $I$ is $n\times n$.  Multiplication by $\lambda^m$ is represented
by $\SSS_k^m$, which has a similar block-Toeplitz structure to
$\SSS_k$ for $2\le m\le s-1$, but with $S$ replaced by $S^m$ and $Z$
replaced by
\begin{equation}
  Z_m = \begin{cases}
    0 & \quad m=0 \\
    \sum_{l=0}^{m-1}S^lZS^{m-1-l},
  & 1\le m\le s-1.
\end{cases}
\end{equation}
The matrix $p(\SSS_k)^j$ is a shift operator with identity
blocks $I_n\otimes I_s$ on the $j$th sub-diagonal.
If we expand
\begin{equation}
  A(\lambda) = A^\ee0 + p A^\ee1 + \cdots + p^q A^\ee q, \qquad
  A^\ee j= A^\ee{j,0} + \cdots + \lambda^{s-1}A^\ee{j,s-1},
\end{equation}
the matrix $\AAA_k$ representing $A(\lambda)$ is given by
\begin{equation}\label{eqn:AAAk:def}
  \AAA_k = \left(\begin{array}{cccc}
      A_0 & 0 & \cdots & 0 \\
      A_1 & A_0 & \cdots & 0 \\
      \cdots & \cdots & \cdots & \cdots \\
      A_k & A_{k-1} & \cdots & A_0
      \end{array}\right),
\end{equation}
where
\begin{equation}\label{eqn:Aj:def}
    A_j = \begin{cases}
      \sum_{m=0}^{s-1} A^\ee{0,m}\otimes S^m, & \quad j=0, \\
      \sum_{m=0}^{s-1} \big[ A^\ee{j,m}\otimes S^m +
      A^\ee{j-1,m}\otimes Z_m\big], & \quad j\ge1.
    \end{cases}
\end{equation}
This formula may be derived by observing that the matrix
representation of the action of $p(\lambda)^j\lambda^m A^\ee{j,m}$ on
$M/p^{k+1}M$ is block Toeplitz with $A^\ee{j,m}\otimes S^m$ on the
$j$th sub-diagonal and $A^\ee{j,m}\otimes Z_m$ on the $(j+1)$st.
Defining $A_j$ this way avoids the need to compute remainders and
quotients in subsequent steps (such as occur in Algorithm~\ref{alg2}).

Next we seek an efficient method of computing a basis matrix $\XX_k$
for the nullspace $\WW_k=\ker\AAA_k$.
Suppose $k\ge1$ and we have computed $\mathbb{X}_{k-1}$.  The first
$k$ blocks of equations in $\AAA_{k}\XX_{k}=0$ imply there are
matrices $\UU_k$ and $\YY_k$ such that
$\XX_{k}=[\XX_{k-1}\UU_k;\YY_k]$, while the last block
of equations is
\begin{equation}\label{eqn:mcA:def}
\overbrace{\left(
    \begin{array}{cccc}
      A_k & \dots & A_0
    \end{array}
  \right)\left(
    \begin{array}{cc}
      0 & \XX_{k-1} \\
      I_{sn\times sn} & 0\\
    \end{array}\right)}^{\mc{A}_{k}}
\hspace*{-1.104in}
\underbrace{\phantom{\left(
      \begin{array}{cc}
        0 & \XX_{k-1} \\
        I_{sn\times sn} & \BB_{k-1} \\
      \end{array}\right)} 
  \left(\begin{array}{c}
      \YY_k \\
      \UU_k \\
    \end{array}
  \right)}_{\XX_{k}}
=\left(
  \begin{array}{c}
    0_{sn\times sR_{k}}
  \end{array}
\right).
\end{equation}
The
matrices $\XX_k$ can be built up recursively by setting
$X_0=\XX_0=\YY_0=\opn{null}(A_0)$,
defining $\UU_0$ to be an empty matrix (with zero rows and $sR_0$
columns), and computing
\begin{align*}
  &\mc{A}_k = \Big( \mc{A}_{k-1} \,,\,
  [A_k,\dots,A_1]X_{k-1} \Big), \qquad (k\ge 1) \\
  &[Y_k;U_k] = \text{new columns of } \opn{null}(\mc{A}_k) \text{
    beyond those of } \opn{null}(\mc{A}_{k-1}), \\
  &[\YY_k;\UU_k] =
  \begin{pmatrix}
    \YY_{k-1} & Y_k \\
    [\UU_{k-1};0] & U_k
  \end{pmatrix}, \qquad
  \begin{aligned}
    X_k &= [\XX_{k-1}U_k; Y_k], \\[2pt]
    \XX_k &= [\iota(\XX_{k-1}),X_k].
  \end{aligned}
\end{align*}
Here $\iota:K^{snl}\rightarrow K^{sn(l+1)}$ represents multiplication
by $p$ from $M/p^lM$ to $M/p^{l+1}M$:
\begin{equation}
  \iota([x^\ee0;\dots;x^\ee{l-1}]) = [0;x^\ee0;\dots;x^\ee{l-1}], \qquad
  x^\ee j,\,0\in K^{sn}.
\end{equation}
By construction, $\XX_k=[\iota(\XX_{k-1}),X_k]$ is a basis for $\WW_k$
when $k\ge1$; it follows that $X_k+\iota(\WW_{k-1})$ is a basis for
$W_k$ when $W_k$ is viewed as a vector space over $K$.  We define
$X_0=\XX_0$ and $X_{-1}=I_{sn\times sn}$ to obtain bases for $W_0$ and
$W_{-1}$ as well.

But we actually want a basis for $W_k$ viewed as a vector space over
$R/pR$ rather than $K$.  Fortunately, all the matrices in this
construction are manipulated $s\times s$ blocks, and the desired basis
over $R/pR$ may be obtained by selecting the first column from each
supercolumn (group of $s$ columns) of~$X_k$.  Indeed, if
$[x_1,\dots,x_s]$ is a supercolumn of $X_k$, we are able to prove that
$x_{j} - \SSS_k x_{j-1} \in \iota(\WW_{k-1})$ for $2\le j\le s$.
Since $\SSS_k$ represents multiplication by $\lambda$, these columns
are all equivalent over $R/pR$.  We are also able to prove that
constructing $X_k$ in this way (using the first column of each
supercolumn) 
is equivalent to Algorithm~\ref{alg2}, i.e.~it yields the same
unimodular matrix $V(\lambda)$ that puts $A(\lambda)$ in local
Smith form.
We omit the proof as it is long and technical, involving a careful
comparison of nullspace calculations via row-reduction in the two
algorithms.

In practice, this version of the algorithm (over $K$) is
easier to implement, but the other version (over $R/pR$) should be
about $s$ times faster as the cost of multiplying two elements of
$R/pR$ is $O(s^2)$ while the cost of multiplying two $s\times s$
matrices is $O(s^3)$.  The results in Section~\ref{sec_numerical} were
computed as described in this appendix (over $K=\QQ$).

\bibliographystyle{plain}
\bibliography{reference}
\end{document}